\documentclass[conference]{IEEEtran}[10 pt]
\ifCLASSINFOpdf
\else
\fi
\usepackage{amssymb}
\usepackage{amsmath}
\usepackage{amsthm}
\usepackage[utf8]{inputenc}
\usepackage[english]{babel}
\usepackage{geometry}
\usepackage{mathtools}
\usepackage{graphicx}
\usepackage{xcolor}

\newtheorem{mylem}{Lemma}
\newtheorem{mythem}{Theorem}

\newtheorem{mydef}{Definition}
\usepackage{float}
\usepackage[bookmarks=false]{hyperref}
\hypersetup{
	colorlinks=true,
	linkcolor=blue,
	filecolor=magenta,      
	urlcolor=cyan,
}

\usepackage[linesnumbered,ruled,vlined]{algorithm2e}

\SetCommentSty{mycommfont}

\usepackage{float}
\usepackage{enumerate}

\begin{document}
%
\title{On the Enumeration of Maximal $(\Delta, \gamma)$\mbox{-}\textit{Cliques} of a Temporal Network}



\author{\IEEEauthorblockN{Suman Banerjee\IEEEauthorrefmark{0}, Bithika Pal\thanks{Both the authors have equally contributed in this work.} \IEEEauthorrefmark{0}}
	
	\IEEEauthorblockA{\IEEEauthorrefmark{0}
		Department of Industrial and Systems Engineering, Indian Institute of Technology, Kharagpur, 721302, India,\\ Email: suman@iitkgp.ac.in, bithikapal@iitkgp.ac.in}

}

%


\maketitle

\begin{abstract}
	A \textit{temporal network} is a mathematical way of precisely representing a time varying relationship among a group of agents. In this paper, we introduce the notion of $(\Delta, \gamma)$\mbox{-}\textit{Cliques} of a temporal network, where every pair of vertices present in the clique communicates atleast $\gamma$ times in each $\Delta$ period within a given time duration. We present an algorithm for enumerating all such maximal cliques present in the network. We also implement the proposed algorithm with three human contact network data sets. Based on the obtained results, we analyze the data set on multiple values of $\Delta$ and $\gamma$, which helps in finding out contact groups with different frequencies.  
\end{abstract}


%
\IEEEpeerreviewmaketitle

\section{Introduction}
A group of agents and \textit{binary relation} among them can be mathematically formalized as a network (also known as graph). Analysis of such network for different topological properties forms the basis of  several domains namely \textit{Social Network Analysis}, \textit{Computational Biology} \cite{hulovatyy2015exploring}, \textit{Epidemiology} \cite{masuda2017temporal}. One such topological property could be maximally connected subgraphs; popularly known as cliques. Finding the clique of maximum cardinality in a network is a well known NP\mbox{-}Complete Problem \cite{garey2002computers}. However, more general problem in network analysis could be not only just finding the maximum size clique, but also enumerate all maximal cliques present in the network. Bron et al. \cite{bron1973algorithm} first proposed an enumeration algorithm for maximal cliques in the network which forms the basis of study on this problem. Later, there were advancements for this problem for different types of networks \cite{cheng2012fast}, \cite{eppstein2011listing}, \cite{eppstein2013listing1} etc.  
\par Most of the real life networks from social to biological are \textit{time varying}, which means that the existence of relationship between any two agents changes with time. Temporal networks \cite{holme2012temporal} (also known as \textit{time varying networks} or \textit{link streams}) are the mathematical tools used for precisely representing these time varying relationships. For such kind of networks, a natural extension of clique is the \textit{temporal clique} which comprises of a set of vertices along with a time duration. 
\subsection{Related Work}
This work is closely related to the \textit{$\Delta$\mbox{-}Clique Enumaration Problem} of a temporal network introduced by Virad et al. \cite{viard2015revealing} \cite{viard2016computing}. For a given value of ($\Delta$, $\gamma$)\mbox{-}Clique is defined as a set of vertices of the network with a time interval, such that every pair of veritices of the set has at least one edge in every $\Delta$ time interval. Based on their proposed algorithm, they analyzed contact relationship among a group of students and showed that it brings a different interpretation in  their communication pattern \cite{viard2015revealing}. Later, Himmel et al. \cite{himmel2016enumerating} adopted \textit{Bron\mbox{-}Kerbosch Algorithm} for maximal clique enumeration and proposed its temporal version. Results reported in \cite{himmel2016enumerating} show that, their algorithm performs much better than that of in \cite{viard2016computing} in terms of worst case computational time analysis as well as in experimentaion with real life data sets. Recently, Mukherjee et al. \cite{mukherjee2017enumeration} \cite{mukherjee2015mining} studied maximal clique enumeration problem of an uncertain graph. They introduced the notion of \textit{$\alpha$\mbox{-}maximal clique} in an uncertain graph and also proposed an enumeration algorithm for all such maximal cliques.
\subsection{Our Contribution}
For a group of human with their time varying relationship represented as a temporal network, its a natural question which set of people contact frequently among themselves? At minimum how many times they contact within a time interval? Motivated by such questions, in this paper, we introduce the notion of $(\Delta, \gamma)$\mbox{-}\textit{Cliques} which is defined as the set of vertices with a time interval where every pair of users of the set has atleast $\gamma$ intersections in each $\Delta$ time interval. Particularly, we make the following contributions in this paper:
\begin{itemize}
	\item We define the problem of ``\textit{enumeration of ($\Delta, \gamma$)\mbox{-}Cliques}'' of a temporal networks.
	\item We propose an algorithm for enumerating all maximal ($\Delta, \gamma$)\mbox{-}Cliques with detailed analysis and theoretical properties.
	\item We implement the proposed algorithm on three human contact network data sets and investigate the deeper insights of contact pattern. 
\end{itemize}
\subsection{Organization of the Paper} 
Rest of the paper is organized as follows: Section \ref{Sec:Preli} describes some preliminary concepts that is required to understand the rest of the paper. In Section \ref{Sec:DG}, we introduce the notion of  $(\Delta, \gamma)$\mbox{-}Clique of a temporal network and its various properties. Section \ref{EMC} contains our proposed enumeration algorithm of all maximal ($\Delta, \gamma$)\mbox{-}Cliques and its detailed analysis. Section \ref{ED} contains experimental details which covers description of datasets,  obtained results from the experiment and their discussions. Finally, Section \ref{CFD} concludes our works and gives future directions.

\section{Preliminaries} \label{Sec:Preli}
In this section we describe some preliminary concepts, which will form the basis to understand the work presented in the subsequent sections of this paper. A \textit{temporal network} is a graph whose edges are associated with a \textit{time stamp} to denote the time at which the edge appeared. Formally, it is defined as follows: 
\begin{mydef}[Temporal Network] \cite{holme2013temporal}
	Temporal network (also known as time varying graphs, link streams) is defined as a triplet $\mathcal{G}(V, E, \mathcal{T})$, where $V(\mathcal{G})$ and $E(\mathcal{G})$ ($E(\mathcal{G}) \subset \binom {V(\mathcal{G})}2 \times \mathcal{T}$) are the vertex and edge set of the network. $\mathcal{T}$ is a function which assigns each edge to its occurrence time stamp.
\end{mydef}
Figure \ref{Fig:TG} shows the time varying links of a temporal network. Suppose, the network $\mathcal{G}$ is observed in discrete time steps (spaced by $dt$) starting from the time $t$ and continued till time $t^{'}$, i.e., $\mathbb{T}=\{t, t+dt, t+2dt, \dots , t^{'} \}$ (suppose, $t^{'}=t+ndt$) Hence, $\mathcal{T} : E(\mathcal{G}) \longrightarrow \mathbb{T}$. Each edge of $\mathcal{G}$ is of the form $(v_i, v_j, t_{ij})$ signifying that there is an edge between the nodes $v_i$ and $v_j$ at time $t_{ij}$. The difference, $t^{'}-t$ is known as the \textit{lifetime} of the network and it is denoted as $\textbf{T}$. In our case $\textbf{T}=ndt$. We say $(v_iv_j)$ is a static edge of $\mathcal{G}$, if $(v_i, v_j, t_{ij}) \in E(\mathcal{G})$ for some $t_{ij} \in \mathbb{T}$. We define the frequency of an edge as the number of times the edge has occurred in the entire lifetime of the network. We define the frequency of the static edge $(v_iv_j)$ as $f_{v_iv_j}$. Viard et al. introduced the notion of $\Delta$\mbox{-}Cliques of a temporal networks  which is a natural extension of cliques in a static network. 
\begin{mydef}[$\Delta$\mbox{-}Clique]
	For a given time period $\Delta$, a $\Delta$\mbox{-}Clique of the temporal network $\mathcal{G}$ is a vertex set, time interval pair, i.e., $(\mathcal{X}, [t_a,t_b])$ with $\mathcal{X} \subset V(\mathcal{G})$, $\vert \mathcal{X} \vert \geq 2$ and $[t_a,t_b] \subset \mathbb{T}$, such that $\forall v_i,v_j \in \mathcal{X}$ and $\tau \in [t_a, max(t_b - \Delta, t_a)]$ there is an edge $(v_i, v_j, t_{ij}) \in E(\mathcal{G})$ with $t_{ij} \in [\tau, min (\tau + \Delta, t_b)]$.
\end{mydef}
\section{$(\Delta, \gamma)$\mbox{-}Clique of a Temporal Network}\label{Sec:DG}
In this section, we introduce the notion of $(\Delta, \gamma)$\mbox{-}Clique of a temporal network
\begin{mydef}[$(\Delta, \gamma)$\mbox{-}Clique]\label{Def:DG}
	For a given time period $\Delta$ and $\gamma \in \mathbb{Z}^{+}$, a $(\Delta, \gamma)$\mbox{-}Clique of the temporal network $\mathcal{G}$ is a vertex set, time interval pair, i.e., $(\mathcal{X}, [t_a,t_b])$ where $\mathcal{X} \subseteq V(\mathcal{G})$, $\vert \mathcal{X} \vert \geq 2$, and  $[t_a,t_b] \subseteq \mathbb{T}$. Here $\forall v_i,v_j \in \mathcal{X}$ and $\tau \in [t_a, max(t_b - \Delta, t_a)]$, there must exist $\gamma$ or more number of edges, i.e., $(v_i, v_j, t_{ij}) \in E(\mathcal{G})$ and $f_{(v_iv_j)} \geq \gamma$ with $t_{ij} \in [\tau, min (\tau + \Delta, t_b)]$.
\end{mydef}

\begin{figure}
	\centering
	\resizebox{6.5 cm}{1.8 cm}{\includegraphics{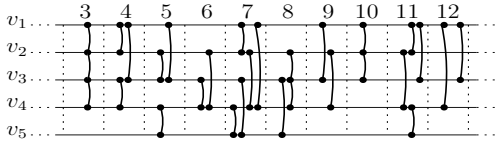}}
	\caption{Links of a Temporal Network}
	\label{Fig:TG}
\end{figure}

\begin{figure}
	\centering
	\resizebox{6.5 cm}{1.8 cm}{\includegraphics{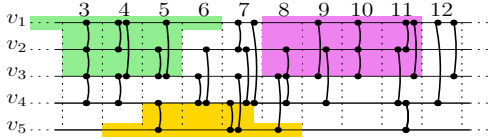}  }
	\caption{Cliques of the Figure \ref{Fig:TG}: $(\{v_1, v_2, v_3\},[2,6])$, $(\{v_1, v_2, v_3\},[8,11])$ and $(\{v_4, v_5\},[4,8])$.}
	\label{Fig:DGClique}
\end{figure}
Figure \ref{Fig:TG} is an example of a time varying network. Figure \ref{Fig:DGClique} shows some $(\Delta, \gamma)$\mbox{-}Cliques for the temporal network of Figure \ref{Fig:TG} with $\Delta=3$ and $\gamma=2$.  Now, from the Definition \ref{Def:DG}, it is easy to observe, that a $(\Delta, \gamma)$\mbox{-}Clique will be a $\Delta$\mbox{-}Clique when $\gamma=1$. 

For a static network $G(V, E)$, a clique $\mathcal{S} \subset V(G)$ is maximal if $v \in V(G) \setminus \mathcal{S}$, $\mathcal{S} \cup \{v\}$ is not a clique. But in case of $(\Delta, \gamma)$\mbox{-}Clique as it is defined in the context of a temporal network, so its maximality has to be decided based on both its \textit{cardinality} and \textit{time interval}. By considering both the factors, we define the maximality condition for an arbitrary $(\Delta, \gamma)$\mbox{-}Clique in Definition \ref{Def:MDG}.
\begin{mydef}[Maximal $(\Delta, \gamma)$\mbox{-}Clique] \label{Def:MDG} \label{Def:maximal}
	A $(\Delta, \gamma)$\mbox{-}Clique $(\mathcal{X}, [t_a,t_b])$ of the temporal network $\mathcal{G}(V, E, \mathcal{T})$ will be maximal if neither of the following is true.
	\begin{itemize}
		\item $\exists v \in V(\mathcal{G}) \setminus \mathcal{X}$ such that $(\mathcal{X} \cup \{v\}, [t_a,t_b])$ is a $(\Delta, \gamma)$\mbox{-}Clique.
		\item $(\mathcal{X}, [t_a - dt,t_b])$ is a $(\Delta, \gamma)$\mbox{-}Clique. This condition is applied only if $t_a - dt \geq t$.
		\item $(\mathcal{X}, [t_a,t_b + dt])$ is a $(\Delta, \gamma)$\mbox{-}Clique. This condition is applied only if $t_b + dt \leq t^{'}$.
	\end{itemize}
\end{mydef}

From the definition of maximal $(\Delta, \gamma)$\mbox{-}Clique it is  easy to observe that first condition  is regarding the cardinality whereas second and third one for time duration. For static graphs, among all the maximal cliques, one whose cardinality is maximum is maximal clique. However, in our case maximum can be in terms of either time duration or cardinality. Hence, maximum $(\Delta, \gamma)$\mbox{-}Clique of temporal network is defined as follows. 
\begin{mydef}[Maximum $(\Delta, \gamma)$\mbox{-}Clique]
	Let $\mathcal{S}$ be the set of all maximal $(\Delta, \gamma)$\mbox{-}Cliques of the temporal network $\mathcal{G}(V, E, \mathcal{T})$. Now, $(\mathcal{X}, [t_a,t_b]) \in \mathcal{S}$ will be
	\begin{itemize}
		\item temporally maximum if $\forall (\mathcal{Y}, [t_a^{'},t_b^{'}]) \in \mathcal{S} \setminus (\mathcal{X}, [t_a,t_b])$, $t_b-t_a \geq t_b^{'} - t_a^{'}$.
		\item cardinally maximum if $\forall (\mathcal{Y}, [t_a^{'},t_b^{'}]) \in \mathcal{S} \setminus (\mathcal{X}, [t_a,t_b])$, $\vert \mathcal{X} \vert \geq \vert \mathcal{Y} \vert$.
	\end{itemize}
\end{mydef}
\section{Enumeration of Maximal $(\Delta, \gamma)$\mbox{-}Clique}\label{EMC}
In this section, we present our proposed enumeration algorithm. The idea is based on the enumeration of maximal cliques in a static graph \cite{johnson1988generating}. It is divided into two parts. First, we initialize all the trivial $(\Delta, \gamma)$\mbox{-}Cliques (Algorithm \ref{Algo:Ini}) and then we try to expand them both in cardinality and temporally till maximality is not reached by Definition \ref{Def:maximal} (Algorithm \ref{Algo:Enum}).

\begin{algorithm} \label{Algo:Ini}
	\SetAlgoLined
	\KwData{A Temporal Graph $\mathcal{G}(V, E, \mathcal{T}), \ \Delta, \ \gamma$.}
	\KwResult{Initial Clique Set.}
	$\text{Prepare the Dictionary }\mathcal{D}$\;
	$\mathcal{C}=\{ \}$\;
	\For{\text{All} $(uv) \in \mathcal{D}.Keys$}{
		\If{$f_{(uv)} \geq \gamma$}{
			$\mathcal{T}_{(uv)}=\text{Time Stamps of } (uv)$\;
			\For{ $i=1$ to $len(\mathcal{T}_{(uv)})-\gamma+1$}{
				\If{$\mathcal{T}_{(uv)}[i+\gamma-1] - \mathcal{T}_{(uv)}[i] =\Delta$}{
					$\mathcal{C}=\mathcal{C} \cup (\{u,v\},[\mathcal{T}_{(uv)}[i],\mathcal{T}_{(uv)}[i+\gamma-1]])$\;
				}
				\If{$\mathcal{T}_{(uv)}[i+\gamma-1] - \mathcal{T}_{(uv)}[i] <\Delta$}{
					$C_1=(\{u,v\},[\mathcal{T}_{(uv)}[i],\mathcal{T}_{(uv)}[i]+\Delta])$\;
					$C_2=(\{u,v\},[\mathcal{T}_{(uv)}[i+\gamma-1]- \Delta,\mathcal{T}_{(uv)}[i+\gamma-1]])$\;
					$\mathcal{C}=\mathcal{C} \cup \{ C_1, C_2  \}$\;
				}
			}          
		}
	}
	\caption{Initialization of the $(\Delta, \gamma)$\mbox{-}Clique}
\end{algorithm}

In Algorithm \ref{Algo:Ini} to initialize the trivial $(\Delta, \gamma)$\mbox{-}Cliques, we create the dictionary $\mathcal{D}$ where the static edges (vertex pairs) are the keys and the time stamps at which they occur are the values. Then, we define an empty set, $\mathcal{C}$, to store all the initial $(\Delta, \gamma)$\mbox{-}Cliques. Next, for each static edge $(uv)$ present in $\mathcal{D}$, if its frequency is atleast $\gamma$, we put the corresponding dictionary values in a list denoted by $\mathcal{T}_{(uv)}$. Now, in $\mathcal{T}_{(uv)}$, let us consider any particular consecutive $\gamma$ occurrences and denote it as $p$. If $p$ appears exactly in $\Delta$ duration, then we add a clique in $\mathcal{C}$ with the vertex pair $\{u,v\}$ and the time duration $[t_a,t_b]$, where $t_a$ and $t_b$ are the first and last appearing time stamps of $p$. Otherwise, if the duration of $p$ is less than $\Delta$, we add two cliques $C_1$ and $C_2$ with the same vertex pair. For $C_1$, $t_a$ is the first appearing time stamp in $p$ and $t_b$ is $t_a+\Delta$, and for $C_2$, $t_b$ is the last occurring time stamp in $p$ and $t_a$ is $t_b-\Delta$ (\textbf{for}\mbox{-}loop from line number 6 to 15). The same procedure is repeated for each consecutive $\gamma$ occurrences for all static edges present in the dictionary $\mathcal{D}$. Now, we make the following observation for the initial $(\Delta, \gamma)$\mbox{-}Cliques formed by Algorithm \ref{Algo:Ini}.

\begin{mylem}\label{Lemma:ini}
	For each clique $(\mathcal{X}, [t_a,t_b]) \in \mathcal{C}$ in Algorithm \ref{Algo:Ini}, the following relations will always hold:
	\begin{enumerate}[(a)]
		\item  $\vert \mathcal{X} \vert = 2$
		\item  $u,v \in \mathcal{X}, \ \ f_{(uv)}=\gamma$
		\item  $t_b - t_a = \Delta$         
	\end{enumerate}	
\end{mylem}
\begin{proof}
	To initialize a $(\Delta, \gamma)$\mbox{-}Clique, it is trivial that the time span of the clique should be of minimum $\Delta$ duration and each pair of vertices of the clique should appear at least $\gamma$ times within each $\Delta$ duration (as per Definition \ref{Def:DG}). As in Algorithm \ref{Algo:Ini}, we are starting with a static edge and picking each consecutive $\gamma$ occurrences, cardinality of the clique will be 2 and frequency of the edges will be exactly $\gamma$. This proves the conditions $(a)$ and $(b)$. Now, the cliques that are added satisfying the equality condition at line number 7 of Algorithm \ref{Algo:Ini}, naturally they are of $\Delta$ duration. However, if a particular $\gamma$ occurrences happen within $\Delta$ time span, we are adding two cliques of exactly $\Delta$ duration. One is forwarding the time span as first appearing time stamp plus $\Delta$ and another is backwarding the time span as last appearing time stamp minus $\Delta$ (at line number 11 and 12). This proves the condition $(c)$. 
\end{proof}

\par Now, we analyze the time requirement for our initialization process. Preparing the dictionary in Line number 1 requires $\mathcal{O}(\underset{(u,v,t) \in E(\mathcal{G})}{\sum}f_{(uv)})$ time. Assuming frequency of all the edges is atleast $\gamma$, inner loop at Line 6 will take $\mathcal{O}(f_{(uv)})$ time. Hence, for processing all the edges to build the set of initial $(\Delta,\gamma)$\mbox{-}Cliques, $\mathcal{C}$, running time is $\mathcal{O}(\underset{(uv) \in E(\mathcal{G})}{\sum}f_{(uv)})$. So, the total time complexity of Algorithm \ref{Algo:Ini} is $\mathcal{O}(\underset{(uv) \in E(\mathcal{G})}{\sum}f_{(uv)} + \underset{(uv) \in E(\mathcal{G})}{\sum}f_{(uv)})=\mathcal{O}(\underset{(uv) \in E(\mathcal{G})}{\sum}f_{(uv)})$. If we sum up the frequencies of all the static edges, we get the number of temporal edges, i.e., the number of triplets $(u,v,t) \in E(\mathcal{G})$. Let, $\vert E(\mathcal{G}) \vert=m$. Hence running time  of Algorithm \ref{Algo:Ini} is mentioned below.
\begin{mylem}
Running time of the initialization process described in Algorithm 1 is $\mathcal{O}(m)$. 
\end{mylem}
\begin{algorithm} \label{Algo:Enum}
	\SetAlgoLined
	\KwData{$\mathcal{G}(V, E, \mathcal{T}), \ \text{Initial Clique Set} \ \mathcal{(C)}, \ \Delta, \ \gamma$.}
	\KwResult{Set of All Maximal $(\Delta, \gamma)$\mbox{-}Clique ($\mathcal{C_R}$). }
	$\mathcal{C_R}=\phi, \ \mathcal{C_I}=\mathcal{C}$\;
	 \While{$\mathcal{C} \neq  \phi$}{
	 	take and remove ($\mathcal{X}$, [$t_a,t_b$]) from $\mathcal{C}$\;
	 	Prepare the Static Graph $G$ for the duration [$t_a,t_b$]\;
	 	$Is\_Maximal = TRUE$\;
	 	\For{\text{All} $v \in N_G(\mathcal{X}) \setminus \mathcal{X}$}{
	 		\If{$(\mathcal{X} \cup \{v\}, [t_a, t_b])$ is a $(\Delta, \gamma)$\mbox{-}Clique}{
	 			$Is\_Maximal = FALSE$\;
	 			\If{$(\mathcal{X} \cup \{v\}, [t_a, t_b]) \notin \mathcal{C_I}$}{
	 				add $(\mathcal{X} \cup \{v\}, [t_a, t_b])$ to $\mathcal{C}$ and $\mathcal{C_I}$\;
	 			}	
	 		}
 	   }

      $t_{al} = max_{u,v \in \mathcal{X}} \ t_{auv}$  \tcp*{Latest first $\gamma^{th}$ occurrence time of an edge in $(\mathcal{X}, [t_a, t_b])$}
      $t_{a^{'}} = t_{al}-\Delta$\;
      \If{$t_{a^{'}} \neq t_a$}{
      	$Is\_Maximal = FALSE$\;
      	\If{$(\mathcal{X}, [t_{a^{'}}, t_b]) \notin \mathcal{C_I}$}{
      		add $(\mathcal{X}, [t_{a^{'}}, t_b])$ to $\mathcal{C}$ and $\mathcal{C_I}$\;
      	}
      }

   $t_{br} = min_{u,v \in \mathcal{X}} \ t_{buv}$  \tcp*{Earliest last $\gamma^{th}$ occurrence time of an edge in $(\mathcal{X}, [t_a, t_b])$}
  $t_{b^{'}} = t_{br}+\Delta$\;
  \If{$t_{b^{'}} \neq t_b$}{
  	$Is\_Maximal = FALSE$\;
  	\If{$(\mathcal{X}, [t_a, t_{b^{'}}]) \notin \mathcal{C_I}$}{
  		add $(\mathcal{X}, [t_a, t_{b^{'}}])$ to $\mathcal{C}$ and $\mathcal{C_I}$\;
  	}
  }
      
      \If{$Is\_Maximal == TRUE$}{
      	add $(\mathcal{X}, [t_a, t_b])$ to $\mathcal{C_R}$\;
      }
      
 }
\caption{Enumeration Algorithm for Maximal $(\Delta, \gamma)$\mbox{-}Cliques}
\end{algorithm}

\begin{figure*}
\centering
\includegraphics[scale=0.7]{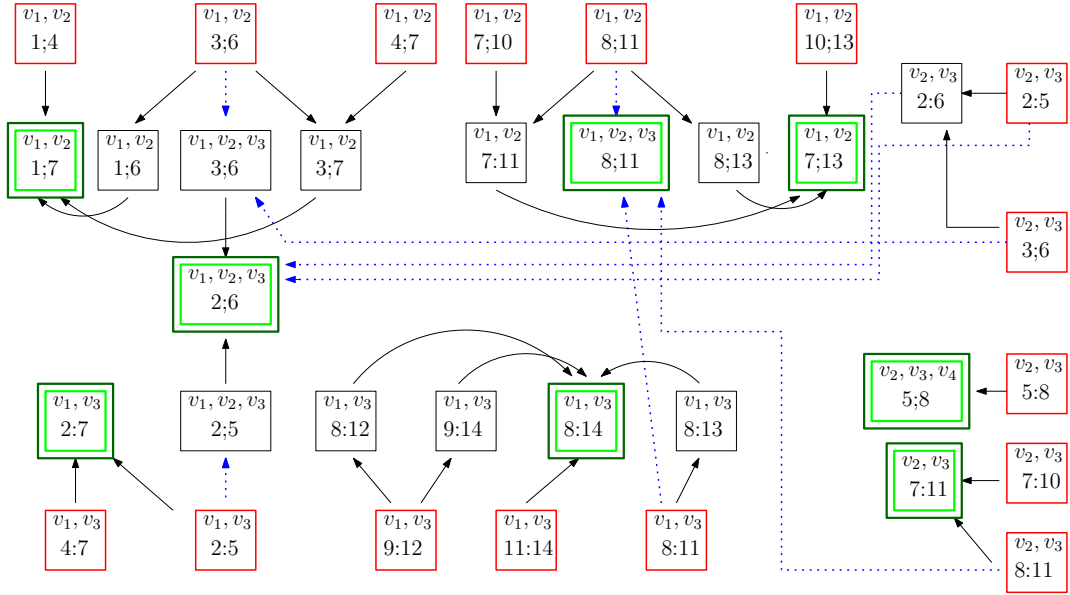}
\caption{Sequence of $(\Delta, \gamma)$\mbox{-}Cliques built by Algorithm \ref{Algo:Enum} for the example mentioned in Figure \ref{Fig:TG} with $\Delta=3$ and $\gamma=2$. We just show the steps for the vertices $v_1$, $v_2$ and $v_3$ only. The red marked boxes denote  the initial $(\Delta, \gamma)$\mbox{-}Cliques returned by Algorithm \ref{Algo:Ini} and green colored boxes denote the maximal $(\Delta, \gamma)$\mbox{-}Cliques. Black solid arrows show time duration expansion and blue dotted arrows show the vertex addition.} 
\label{Fig:Demo}
\end{figure*}

Algorithm \ref{Algo:Enum} describes the enumeration procedure of maximal $(\Delta, \gamma)$\mbox{-}Cliques. For this process, three clique sets $\mathcal{C}$ (for holding the cliques yet to be processed), $\mathcal{C}_{\mathcal{I}}$ (for keeping the cliques already or yet to be processed) and $\mathcal{C}_{\mathcal{R}}$ (for storing the maximal cliques) are maintained. Now, it works as follows. First, it takes out one clique $C_i=(\mathcal{X}, [t_a, t_b])$ from $\mathcal{C}$ and tries to expand (if possible) in any of the following three ways.
\begin{itemize}

\item First one is the addition of nodes (from $N_G(\mathcal{X})$) which is not currently in $\mathcal{X}$ (from Line number 6 to 13). If the addition of a node forms a $(\Delta, \gamma)$\mbox{-}Clique within $t_a$ to $t_b$, we add the new clique to $\mathcal{C}$ and $\mathcal{C}_{\mathcal{I}}$ for further processing and set the Is\_Maximal flag to False, so that $C_i$ can not be added to $\mathcal{C}_{\mathcal{R}}$.

\item Second one is the expansion of the duration in the right side of the time horizon (from Line number 22 to 29). For this expansion, we choose to progress by $\Delta$ duration only, in one iteration. Now, the question arises, from where the $\Delta$ is to be added. If $\Delta$ is added to $t_b$, there is no guarantee that $\forall u,v \in \mathcal{X}$, there will be $\gamma$ edges in $[t_b, t_b+\Delta]$. So, the way is to find the time where all the edges have occurred at least $\gamma$ times from the end ($t_b$) towards $t_a$. For this purpose, we use $t_{buv}$, a list, to store the time on which an edge $(u,v)$ has occurred $\gamma^{th}$ time from last. If we take the minimum (earliest) of $t_{buv}$ as $t_{br}$, that ensures, from $t_{br}$ to $t_b$, there are $\gamma$ edges $\forall u,v \in \mathcal{X}$. Now, if $t_{br}+\Delta > t_b$, the expansion is possible by the definition of $(\Delta, \gamma)$\mbox{-}Clique. The new clique $(\mathcal{X}, [t_a, t_{b^{'}}])$ (where $t_{b^{'}}=t_{br}+\Delta$) is added to $\mathcal{C}$ and $\mathcal{C}_{\mathcal{I}}$ for further processing and the Is\_Maximal flag is set to be False, so that $C_i$ can not be added to $\mathcal{C}_{\mathcal{R}}$.

\item Third one is the expansion of the duration in the left side of the time horizon (from Line number 14 to 21). Similar to the second case, a list, $t_{auv}$, is taken to keep the time at which a static edge $(u,v)$ has occurred $\gamma^{th}$ time from the first ($t_a$). Now, the maximum (latest) of $t_{auv}$ (i.e., $t_{al}$) tells that from $t_a$ to $t_{al}$ there is $\gamma$ occurrences $\forall u,v \in \mathcal{X}$. If $t_{al}-\Delta < t_a$, the new clique $(\mathcal{X}, [t_{a^{'}}, t_b])$ (where $t_{a^{'}}=t_{al}-\Delta$) is added as follows from above.

\end{itemize} 
This process is iterated until $\mathcal{C}$ is empty and finally $\mathcal{C}_{\mathcal{R}}$ contains all the maximal $(\Delta, \gamma)$\mbox{-}Cliques. Figure \ref{Fig:Demo} describes the procedure for building the maximal $(\Delta, \gamma)$\mbox{-}Cliques using Algorithm \ref{Algo:Ini}, \ref{Algo:Enum} for the vertices $v_1, v_2, v_3$ of the temporal network shown in Figure \ref{Fig:TG}.

$\blacktriangleright$ \textbf{Remarks:} One important point to mention here is while building the list $t_{buv}$, we are checking the next time stamp as well. For $u,v \in \mathcal{X}$, let $l_{uv}$ denotes the last $\gamma^{th}$ occurrence in $[t_a, t_b]$. Now, if there are $\gamma$ occurrences of $(u,v)$ between $l_{uv}+1$ to $l_{uv}+1+\Delta$, then $l_{uv}+1$ is kept in $t_{buv}$ instead of $l_{uv}$. This scenario can happen when $\Delta=\gamma+1$. Consider the temporal network shown in Figure \ref{Fig:TG}. When $\Delta=3,\ \gamma=2$, an initial clique $C_i=(\{v_1, v_2\}, [7,10])$ is generated by Algorithm \ref{Algo:Ini}. Now, for making $t_{buv}$, if we do not check from 8 ($l_{uv} + 1 $) to 11, $C_i$ can not be expanded from [7,10] to [7,11]. The same is checked for making $t_{auv}$ to see one time stamp backward.

\begin{mylem} \label{Lemma:Cdelga}
In Algorithm \ref{Algo:Enum}, the contents of $\mathcal{C}$ are $(\Delta,\gamma)$ Cliques.
\end{mylem}
\begin{proof}
We prove this statement by induction hypothesis on the iterations of the \textit{while} loop (from Line 2 to 33). By Lemma \ref{Lemma:ini}, initially the contents of $\mathcal{C}$ are $(\Delta,\gamma)$ Cliques. Now, we assume that, the contents of $\mathcal{C}$ at the end of the $i$\mbox{-}th iteration are $(\Delta,\gamma)$ Cliques. An existing clique, $C_{i}=(\mathcal{X},[t_a,t_b])$ of $\mathcal{C}$ may be modified in Line number 10, 19 and 27 of Algorithm \ref{Algo:Enum}. As per the condition imposed in Line number 7, if $C_{i}$ is modified in Line number 10 with new vertex addition, it results to a $(\Delta,\gamma)$ clique. Now, let us show that at Line number 19 $(\mathcal{X},[t_{al}-\Delta,t_{b}])$ is a $(\Delta,\gamma)$ Clique, where $t_{al}$ is  obtained from Line number 14. As, $C_i$ is a $(\Delta,\gamma)$\mbox{-}Clique, all possible static edges formed by the vertices in $\mathcal{X}$ occur atleast $\gamma$ times in every $\Delta$ duration from $t_{al}$ to $t_b$, where $t_{al}\geq t_{a}$. Moreover, since $t_{al}$ is the latest first $\gamma$\mbox{-}th occurrence time of an edge in $C_{i}$, for all $u,v \in \mathcal{X}$, there is necessarily $\gamma$ edges $(u,v,t)$ in $E(\mathcal{G})$ with $t_a\leq t \leq t_{al}$. Here, $t_{al} \leq t_a + \Delta$, otherwise $(\mathcal{X},[t_a,t_b])$ would not be a $(\Delta,\gamma)$\mbox{-}Clique. Therefore, an edge $(uv)$ occurs atleast $\gamma$ times between $t_{al}-\Delta$ and $t_{al}$, $ \forall u,v \in \mathcal{X}$. Finally, $\mathcal{X},[t_{al}-\Delta,t_{b}]$ is a $(\Delta,\gamma)$ Clique. 
\par Similar argument holds for Line number 22. Hence, clique modified in Line number 27 will also be a $(\Delta,\gamma)$ Clique. So, at the end of $(i+1)$\mbox{-}th iteration the all the cliques in $\mathcal{C}_{\mathcal{R}}$ are $(\Delta,\gamma)$ Cliques. This completes the proof. 
\end{proof}

\begin{figure*}[!ht]
\centering
\includegraphics[scale=0.8]{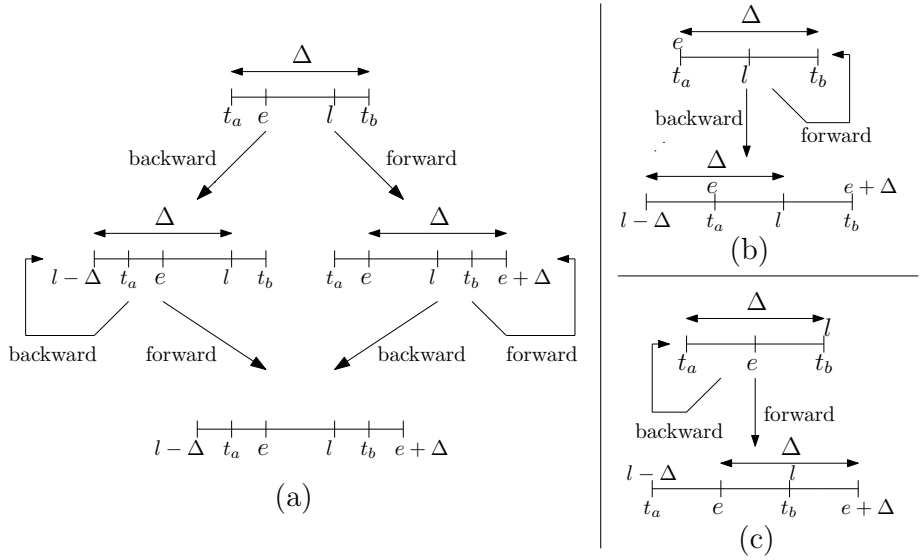} 
\caption{Diagram of expansion in time duration for a $(\Delta,\gamma)$\mbox{-}Clique of Lemma \ref{Lemma:iniall}. (a) represents Case-iv, (b) represents Case-iii and, (c) represents Case-ii.} 
\label{Fig:LD}
\end{figure*}

\begin{mylem} \label{Lemma:TH1}
All the elements of the set $\mathcal{C}_{\mathcal{R}}$ returned by Algorithm \ref{Algo:Enum} are maximal $(\Delta,\gamma)$\mbox{-}Cliques.
\end{mylem}
\begin{proof}
We prove this statement by contradiction. Let us assume, $C_{i}= (\mathcal{X},[t_a,t_b])$ be an element of $\mathcal{C}_{\mathcal{R}}$, which is not maximal. Now, as $\mathcal{C}_{\mathcal{R}}$ can only be filled by the elements of $\mathcal{C}$, so, $C_{i}$ is already a $(\Delta,\gamma)$\mbox{-}Clique (By Lemma \ref{Lemma:Cdelga}). If, $C_{i}$ is not maximal, any one of the following three can happen: 
\begin{enumerate}
\item $\exists v \in N_G(\mathcal{X}) \setminus \mathcal{X}$, such that $(\mathcal{X} \cup \{v\}, [t_a, t_b])$ is a $(\Delta, \gamma)$\mbox{-}Clique. So, $Is \_Maximal$ flag at Line number 8 will be set to false and, $C_{i}$ can not be added to $\mathcal{C}_{\mathcal{R}}$ at Line number 31. Hence, the assumption $C_{i} \in \mathcal{C}_{\mathcal{R}}$ is a contradiction.
\item Let us assume, $\exists t_{a^{'}} < t_a$, such that, $(\mathcal{X},[t_{a^{'}},t_{b}])$ is a $(\Delta,\gamma)$\mbox{-}Clique. Without loss of generality, we assume that there is no edge between the nodes of $\mathcal{X}$ from $t_{a^{'}}$ to $t_{a}$. Let us also consider the latest first $\gamma$\mbox{-}th occurrence time of an edge in $C_i$ is $t_{al} \in [t_a,t_b]$ calculated in Line number 14. So, it is necessary that $t_{al} \leq t_{a^{'}} + \Delta$, as $(\mathcal{X},[t_{a^{'}},t_{b}])$ is a $(\Delta,\gamma)$\mbox{-}Clique. Since, $t_{a^{'}}< t_a$ implies $t_{al}<t_a+\Delta$. As a consequence, if condition at Line number 16 is satisfied and the flag $Is \_ Maximal$ is set to false in Line number 17. Hence, we reach to a contradiction as above.
\item If $\exists t_{b^{'}}>t_b$ such that $(\mathcal{X},[t_{a},t_{b^{'}}])$ is a $(\Delta,\gamma)$\mbox{-}Clique, then in similar way, we reach contradiction as Case 2. 
\end{enumerate}
Finally, $\mathcal{C}_{\mathcal{R}}$ contains only the maximal $(\Delta,\gamma)$\mbox{-}Cliques, which proves the statement. 
\end{proof}
\begin{mylem} \label{Lemma:5}
For a maximal $(\Delta,\gamma)$\mbox{-}Clique $C_i=(\mathcal{X},[t_a,t_b])$,\\ 
(a) if $e$ is the earliest occurrence time of an edge in $C_{i}$, then  $t_b \geq e+\Delta$ \\
(b) if $l$ is the last occurrence time of an edge in $C_{i}$, then $t_a \leq l-\Delta$. 
\end{mylem}
\begin{proof}
We prove this statement by contradiction. Since, $C_{i}$ is a $(\Delta,\gamma)$\mbox{-}Clique, $\forall u,v \in \mathcal{X}$, there exist atleast $\gamma$ number of edges $(u,v,t)$, such that $e\leq t \leq t_b$. Let us assume, $t_b < e+\Delta$. Then, $\forall u,v \in \mathcal{X}$, there also exist $\gamma$ number of edges $(u,v,t)$, such that $e\leq t \leq t_b <e+\Delta$. Hence, $(\mathcal{X},[t_a,e+\Delta])$ is a $(\Delta,\gamma)$\mbox{-}Clique. This implies that $C_i$ is not maximal as $[t_a,t_b] \subset [t_a,e+\Delta]$. Thus, we reach the contradiction. This proves the Statement (a).

In similar way, the Statement (b) can also be proved by assuming $t_a >l-\Delta$.
\end{proof}
\begin{mylem} \label{Lemma:iniall}
In Algorithm \ref{Algo:Ini}, the set $\mathcal{C}$ contains all the required $(\Delta, \gamma)$\mbox{-}Cliques with the properties mentioned in Lemma \ref{Lemma:ini} for Algorithm \ref{Algo:Enum}.
\end{mylem}
\begin{proof}
 Let us assume, that $C_i=(\{u,v\},[t_a,t_a+\Delta])$ be an initial $(\Delta, \gamma)$\mbox{-}Clique. This implies that the static edge $(u,v)$ occurs exactly $\gamma$ times within the time duration $t_a$ to $t_a+\Delta$ (By Lemma \ref{Lemma:ini}). Assume that for these consecutive $\gamma$ occurrences, $(uv)$ appears first at time $e$ and last at time $l$, where $t_a \leq e < l \leq t_a+\Delta$. Now, any one of the following four cases can happen. (i) If $e=t_a$ and $l=t_a+\Delta$, then this clique will be added in Line number 8 of Algorithm \ref{Algo:Ini}. (ii) If $e>t_a$ and $l=t_a+\Delta$, then this clique will be added in Line number 12 of Algorithm \ref{Algo:Ini}. (iii) If $e=t_a$ and $l<t_a+\Delta$, then this clique will be added in Line number 11 of Algorithm \ref{Algo:Ini}. (iv) If $t_a < e < l < t_a+\Delta$, no cliques will be added by Algorithm \ref{Algo:Ini}. Now, we want to show, the cliques of case (iv) are actually redundant. If such a clique, $(\{u,v\},[t_a,t_b])$, where $t_b = t_a + \Delta$ and $t_a < e < l < t_b$, is added in $\mathcal{C}$, then an intermediate clique $(\{u,v\},[l-\Delta,e+\Delta])$ will be reached  by Algorithm \ref{Algo:Enum} at some stage by Line number 14 to 29. Now, if there is consecutive $\gamma$ occurrences of $(u,v)$ within $\Delta$ duration, Algorithm \ref{Algo:Ini} adds two cliques $(\{u,v\},[l-\Delta,l])$ and $(\{u,v\},[e,e+\Delta])$ in $\mathcal{C}$ (by case (ii) and (iii) respectively). These two cliques also form an intermediate clique $(\{u,v\},[l-\Delta,e+\Delta])$ in Algorithm \ref{Algo:Enum}. This implies that case (iv) cliques are redundant and will not effect the computation of all Maximal cliques. Figure \ref{Fig:LD} explains the same. Hence, the statement is proved.
\end{proof}
%

%
%
%

\begin{mylem} \label{Lemma:7}
Let, $C_i=(\mathcal{X}, [t_a,t_b])$ be a maximal $(\Delta,\gamma)$\mbox{-}Clique and $e$ be the earliest occurrence time of an edge in $C_i$. If, $(\mathcal{X}, [e,e+\Delta])$ in $\mathcal{C}$ at some stage of Algorithm \ref{Algo:Enum}, then $C_i \in \mathcal{C}_{\mathcal{R}}$.
\end{mylem}
\begin{proof}
Assume, $C_0=(\mathcal{X}, [e,e+\Delta])$ is in $\mathcal{C}$. We consider the sequence of steps of Algorithm \ref{Algo:Enum} of the form: $C_0 \rightarrow C_1 \rightarrow C_2 \rightarrow \dots \rightarrow C_q$ such that $\forall p, \ C_p=(\mathcal{X}, [e, l_p])$ with $l_{p+1} > l_p$, i.e., the Algorithm \ref{Algo:Enum} builds $C_{p+1}$ from $C_p$ in Line number 22 to 29. Notice that, $t_b \geq e+\Delta$ from Lemma \ref{Lemma:5} and \ref{Lemma:iniall}.
\par We show that $C_q = (\mathcal{X}, [e,t_b])$. If the statement is not true, then, in $C_q = (\mathcal{X}, [e,l_q])$ from sequence, $l_q \neq t_b$. Now, as $C_i$ is maximal, we must have $l_q < t_b$. In addition, $l_q = t_{br} +\Delta$ where $t_{br}$ is the earliest last $\gamma^{th}$ occurrence time of an edge in $C_{q-1}$ computed at Line number 22. Since $C_q$ is the last $(\Delta,\gamma)$\mbox{-}Clique of the sequence, $t_{br}$ is also the earliest last $\gamma^{th}$ occurrence time of an edge in $C_q$; else there will be a clique $C_{q+1}$ satisfying the constraints of the sequence above. Hence, $\exists u,v \in \mathcal{X}$ and $(t_{br}, u, v) \in E(\mathcal{G})$ for which there is no $\gamma$ number of occurrences for $(u,v)$ from $t_{br}+dt$ to $t_{br}+dt+\Delta$. This ensures, $t_{br} + \Delta = t_b$, which implies $l_q = t_b$.
\par Now, we want to show that $C_i$ is constructed from $C_q$, which means $e$ will be expanded towards $t_a$ in some future steps of Algorithm \ref{Algo:Enum}. As, $C_0$ is a $(\Delta,\gamma)$\mbox{-}Clique, there is at least $\gamma$ occurrences $\forall u,v \in \mathcal{X}$. So, the same holds for $C_q$ as well. Now, let $t_{al}$ is the latest first $\gamma^{th}$ occurrence time in $[e,t_b]$, then Algorithm \ref{Algo:Enum} will build $C_i$ with $t_{al}-\Delta = t_a$ from Line number 14 to 22.
\end{proof}
\begin{mylem} \label{Lemma:PreTh}
$\mathcal{C}_{\mathcal{R}}$ contains all the maximal $(\Delta,\gamma)$\mbox{-}Cliques of the temporal network $\mathcal{G}$. 
\end{mylem}
\begin{proof}
As, $\mathcal{C}_{\mathcal{R}}$ is constructed from the clique set $\mathcal{C}$, we need to show that all the maximal $(\Delta,\gamma)$\mbox{-}Cliques are in $\mathcal{C}$ at some iteration of the \textit{while} loop in Algorithm \ref{Algo:Enum}. By lemma \ref{Lemma:iniall}, $\mathcal{C}$ initially contains all the required $(\Delta,\gamma)$\mbox{-}Cliques. Let, $C_i=(\mathcal{X}, [t_a,t_b])$ be a maximal $(\Delta,\gamma)$\mbox{-}Clique, $e$ be the earliest occurrence time of an edge in $C_i$ and let $u,v \in \mathcal{X}$ be two nodes, such that there exists $\gamma$ edges in $E(\mathcal{G})$, starting from $e$. We show that there is a sequence of steps that builds $C_i$ from $(\Delta,\gamma)$\mbox{-}Clique $C_0=(\{u,v\},[e,e+\Delta])$ (which is in $\mathcal{C}$ from Algorithm \ref{Algo:Ini}). Notice that, Algorithm \ref{Algo:Enum} iteratively adds all the elements of $\mathcal{X} \setminus \{u,v\}$ in $C_0$ from Line number 6 to 13. This creates $(\Delta,\gamma)$\mbox{-}Clique $C^{'}=(\mathcal{X},[e,e+\Delta])$ from $C_0$. We finally apply Lemma \ref{Lemma:7} to conclude that the Algorithm \ref{Algo:Enum} builds $C_i$ from $C^{'}.$ 
\end{proof}
By Lemma \ref{Lemma:TH1} and \ref{Lemma:PreTh}, we get the correctness result of Algorithm \ref{Algo:Enum}.
\begin{mythem}
Given a temporal network $\mathcal{G}$ with duration $\Delta$ and edge occurrences $\gamma$, Algorithm \ref{Algo:Enum} correctly enumerates all the maximal $(\Delta,\gamma)$\mbox{-}Clique. 
\end{mythem}
Now, we investigate the time and space requirement of our proposed methodology in worst case. Assume, that the number of nodes and edges of the temporal network is $n$ and $m$ respectively. The worst case will occur, when $\forall u,v \in \mathcal{X}$, for any $\Delta$ duration of the time horizon, there are $\gamma$ edges. In this case, the number of initial cliques returned by  Algorithm $\ref{Algo:Ini}$ will be of $\mathcal{O}(m-\gamma+1)$. Now, as in every $\Delta$ duration there is $\gamma$ edges, so all the initial cliques of $\mathcal{C}$ will be merged into a clique, where the duration will be the entire time horizon. The space consumed by the intermediate cliques of this process is $\mathcal{O}((m-\gamma+1)^{2})$. Now, the number of possible vertex subsets of the temporal graph is of $\mathcal{O}(2^{n})$. If there is a situation, when all $u,v \in \mathcal{X}$, the static edge $(u,v)$ has the frequency $\gamma$ within each duration $\Delta$, the number of intermediate cliques will be $\mathcal{O}(2^{n}(m-\gamma+1)^{2})$. Now, each clique can occupy a space of $\mathcal{O}(n)$. Hence, the space required by this process is $\mathcal{O}(2^{n}n(m-\gamma+1)^{2})$ in worst case.
\par We estimate the time complexity considering the number of basic operations performed by Algorithm \ref{Algo:Enum}. It consists of mainly three blocks; (i)from line number 6 to 13, (ii) 14 to 21, and (iii) 22 to 29. Now, the complexity of the block (iii) is same as (ii). So, we focus on first two blocks. 
\par For a vertex $v \notin \mathcal{X}$, line number 7 tests whether $\mathcal{X} \cup \{v\}$ is a $(\Delta,\gamma)$\mbox{-}Clique or not. To accomplish this task, for each node in $\mathcal{X}$ it has to search all the edges induced by the vertices in $\mathcal{X}$ in the worst case. Hence, required time is $\mathcal{O}(\vert \mathcal{X}\vert . m) \simeq \mathcal{O}(nm)$. At line number 9, it has to search $\mathcal{O}(2^{n}(m-\gamma+1)^{2})$ number of cliques in $\mathcal{C}_{\mathcal{I}}$ and compares two cliques with $\mathcal{O}(n)$ time. Therefore, total time for checking belongingness in $\mathcal{C}_{\mathcal{I}}$ is $\mathcal{O}(n.\log(2^{n}(m-\gamma+1)^{2})) = \mathcal{O}(n^2 +n\log(m-\gamma+1))$. This process is repeated for all $v \in N_G(\mathcal{X}) \setminus \mathcal{X}$ at line number 6 and required time is of $\mathcal{O}(n(nm + n^2 +n\log(m-\gamma+1))) = \mathcal{O}(n^2m+n^3 +n^2\log(m-\gamma+1)) = \mathcal{O}(n^2m+n^3)$.
\par Computing $t_{al}$ at line number 14 requires $\mathcal{O}(m)$ time. Line number 26 also takes $\mathcal{O}(n^2 +n\log(m-\gamma+1))$ time for checking belongingness in $\mathcal{C}_{\mathcal{I}}$. So, the time complexity of the block (ii) is $\mathcal{O}(m + n^2 +n\log(m-\gamma+1))$. 
\par At last, one iteration of the while loop costs $\mathcal{O}(n^2m+n^3 + m + n^2 +n\log(m-\gamma+1)) = \mathcal{O}(n^2m+n^3 )$. Now, the while loop runs for $\mathcal{O}(\vert \mathcal{C}_{\mathcal{I}} \vert) = \mathcal{O}(2^{n}(m-\gamma+1)^{2})$. Hence, the overall time complexity of the proposed methodology is of $\mathcal{O}(2^{n}(m-\gamma+1)^{2}(n^2m+n^3 )) = \mathcal{O}(2^{n}m^3n^2+2^{n}n^3m^2)$.

From this analysis we obtain the following result.
\begin{mythem}
Given a temporal network $\mathcal{G}(V,E,\mathcal{T})$ with $\vert V \vert$ = n, $\vert E \vert$ = m, for enumerating all the maximal $(\Delta,\gamma)$\mbox{-}Cliques the proposed methodology takes $\mathcal{O}(2^{n}n(m-\gamma+1)^{2})$ space and $\mathcal{O}(2^{n}m^3n^2+2^{n}n^3m^2)$ time. 
\end{mythem}

\section{Experimental Details}\label{ED}
In this section, we describe the experimental evaluation of our proposed methodology. We start with a brief description of the datasets that we use in our experiment.
\subsection{Dataset Description}

\begin{itemize}
	\item \textbf{Infectious} \cite{konect:2017:sociopatterns-infectious}, \cite{konect:sociopatterns}: This network describes the face-to-face behavior of people during the exhibition INFECTIOUS: STAY AWAY in 2009 at the Science Gallery in Dublin. Nodes represent exhibition visitors; edges represent face-to-face contacts that were active for at least 20 seconds. Multiple edges between two nodes are possible and denote multiple contacts. The network contains the data from the day with the most interactions.
	\item \textbf{Haggle} \cite{chaintreau2007impact}: This undirected network represents contacts between people measured by carried wireless devices. A node represents a person and an edge between two persons shows that there was a contact between them.
	\item \textbf{College Message} \cite{panzarasa2009patterns}:This dataset is comprised of private messages sent on an online social network at the University of California, Irvine. Users could search the network for others and then initiate conversation based on profile information. An edge $(u, v, t)$ means that user $u$ sent a private message to user $v$ at time $t$.
\end{itemize} 
\par We brief a preliminary statistics of these datasets in Table  \ref{Tab:1}.
\begin{table}
\begin{center}
\caption{Basic statistics of the datasets used in the experiment.}
    \begin{tabular}{ | p{2.0cm} | p{0.6cm} | p{1.8cm} | p{1.4cm} | p{1.1cm} |}
    \hline
    Dataset Name & Nodes & Temporal Links & Static Edges & Lifetime \\ \hline
    Infectious & 410 & 17298 & 2765 & 8 Hours  \\ \hline
    Haggle & 274 & 28244 & 2899 & 4 days \\ \hline
    College Message & 1899 & 59835 & 20296 & 193 Days \\
    \hline
    \end{tabular}
    \label{Tab:1}
\end{center}
\end{table}
\subsection{Goal of the Experiment}
In the context of human contacts, a ($\Delta$, $\gamma$)\mbox{-}Clique signifies a group of frequently interacted persons (set of vertices of the clique) for a particular time length (duration of the clique). From the experimentation, we want to study the following facts:
\begin{itemize}
\item Number of frequently contacted groups by varying their contact interval ($\Delta$) [Count of maximal cliques].
\item Maximum duration of contacts of the frequently contacted groups [Maximum duration among all the maximal cliques].
\item Maximum number of persons from the frequently contacted groups [Maximum cardinality among all the maximal cliques].
\item From the computational framework, number of iterations to obtain the frequently contacted groups. Instead of CPU time, we report number of iterations, which is platform independent.
\end{itemize}
\subsection{Experimental Setup}
Here, we describe our experimental set up. The main two parameters in this problem are $\Delta$ and $\gamma$. It is natural that one will be interested in finding the set of frequently contacted users for a time duration comparable with the network lifetime. Hence, we select the value of $\Delta$ for each dataset based on their lifetime. For the `Infectious' data set, we select the initial value of $\Delta$ as $1$ minute, increment by $1$ minute and continued till $10$ minutes. For the `Haggle' dataset, we adopt similar set up as Infectious. For the `College Message' data set, we initially set the $\Delta$ value as $1$ hour and then $12$, $24$, $72$ and $168$ hours. In case of Infectious and Haggle dataset, for each delta value, we start with $\gamma=2$ and increment by $1$, till the maximal clique set is null. As the chosen $\Delta$ values for the College Message dataset are larger, we increment the $\gamma$ by $5$ till it reaches to $20$ and by $10$ till the maximum clique set is null. The reason behind larger increment is for obtaining a significant change in the result (e.g., maximum cardinality, maximum duration etc.). We run our experiments on a server having Intel Xeon, 2.2 GHz, 16 core processor and 64 GB memory.
\subsection{Results and Observations}
Here, we present the obtained results and observations from the experiment. For Infectious dataset, Figure \ref{Fig:results}-(i) shows the results for the change in number of maximal cliques (i-a), maximum clique duration from the maximal cliques (i-b), maximum cardinality (i-c) and number of iterations (i-d) with the variation of $\Delta$ and $\gamma$ values. The same for Haggle and College Message dataset are shown in Figure \ref{Fig:results}-(ii) and (iii) respectively. Now, we discuss the observations with respect to the  mentioned metrics below.

\begin{figure*}[!htbp]
\centering
\begin{tabular}{ccc}
\textbf{\underline{Infectious}} & \textbf{\underline{Haggle}} & \textbf{\underline{College Message}} \\
\includegraphics[width=6cm,height=5.5cm]{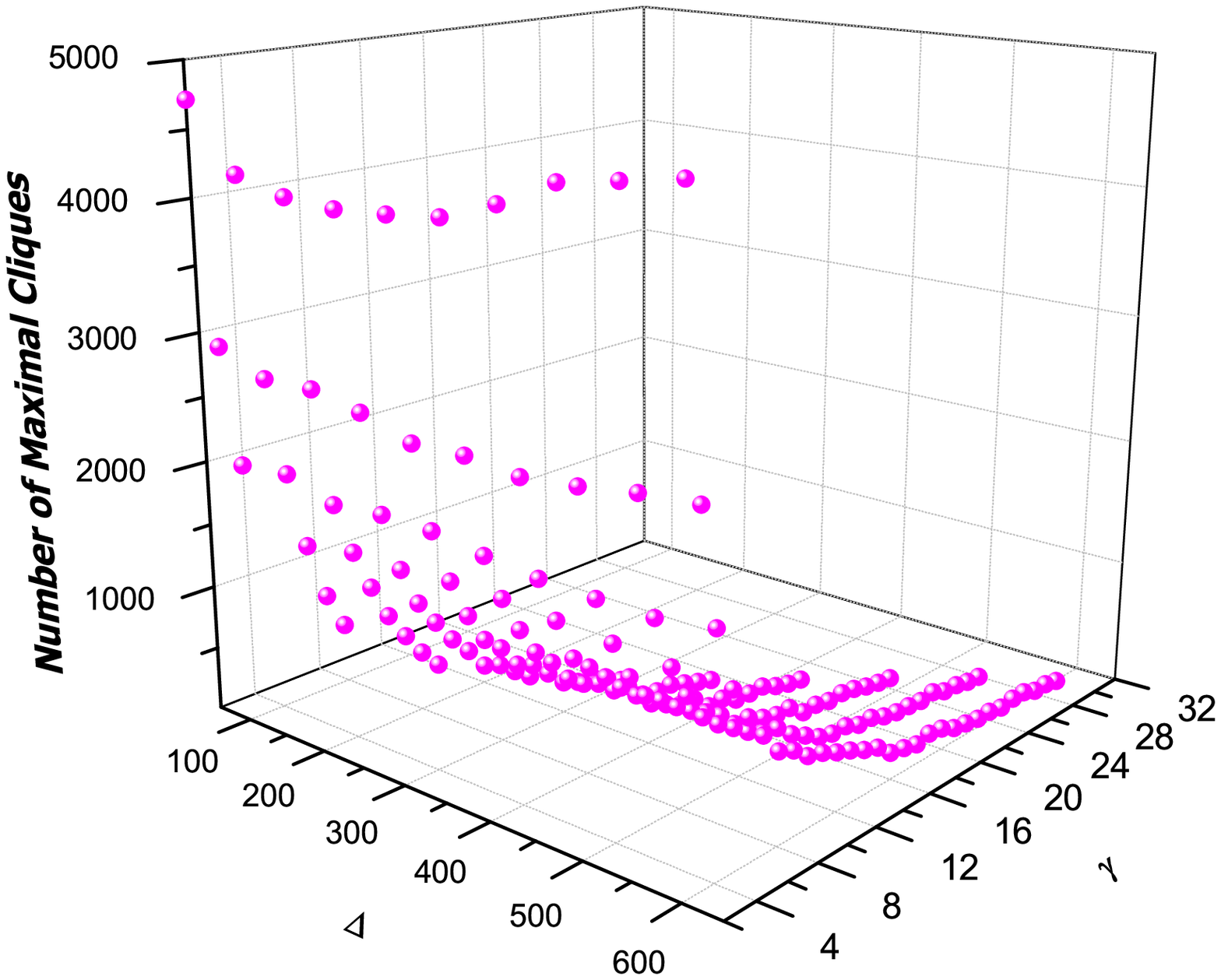} 
 &
\includegraphics[width=6cm,height=5.5cm]{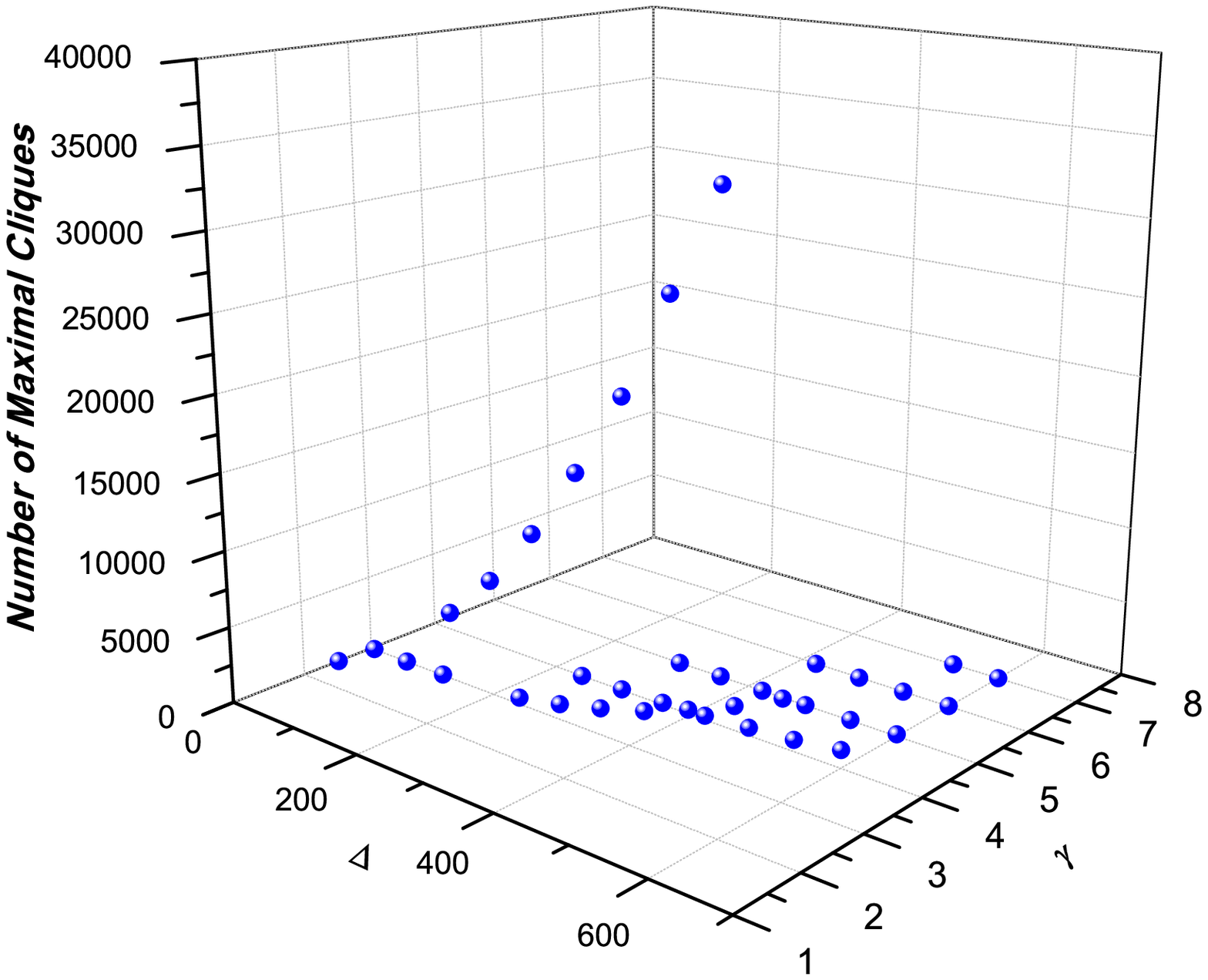}
&
\includegraphics[width=6cm,height=5.5cm]{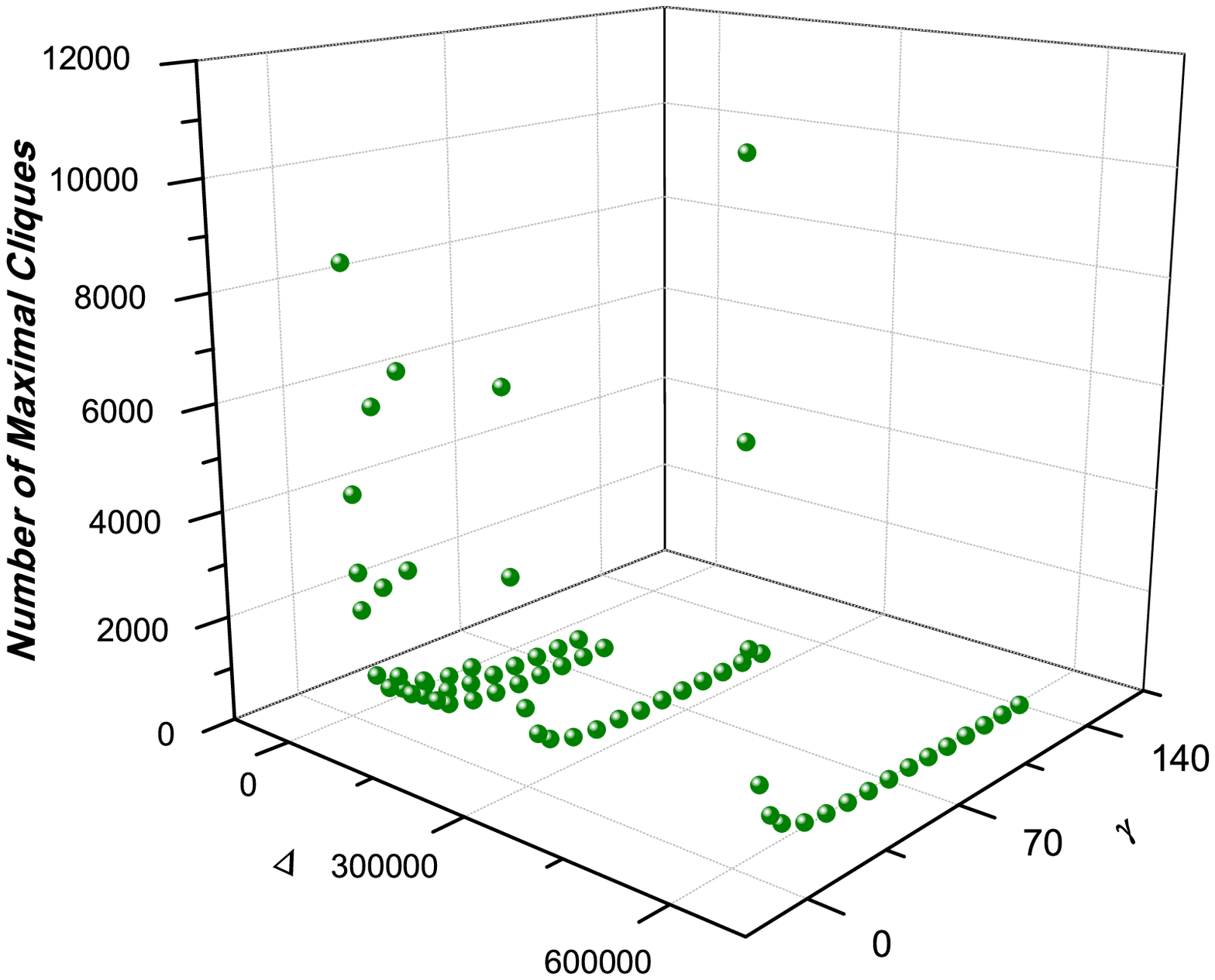}
\\
(i-a)& (ii-a)&(iii-a)\\

 \includegraphics[width=6cm,height=5.5cm]{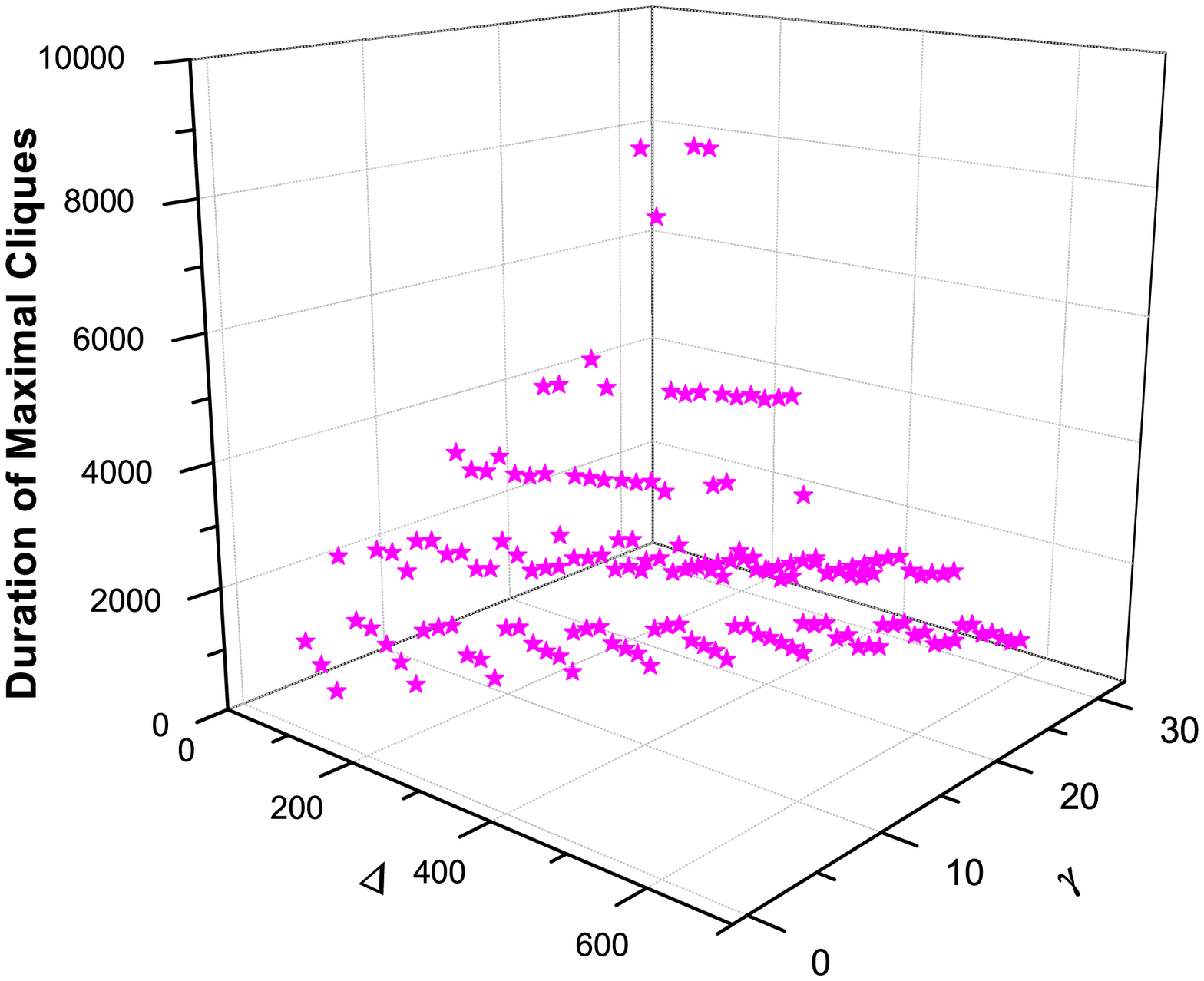} 
 &
 \includegraphics[width=6cm,height=5.5cm]{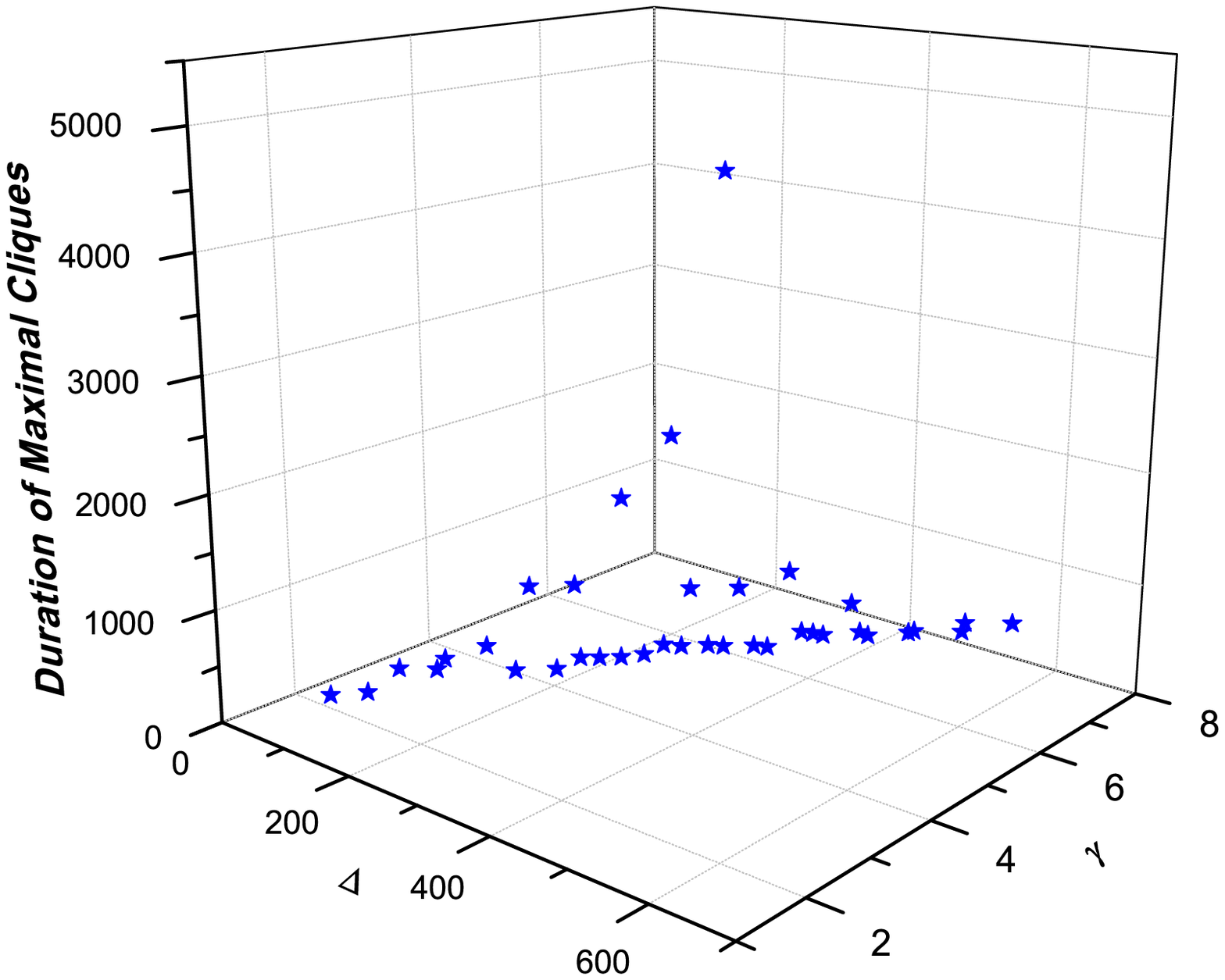} 
 &
 \includegraphics[width=6cm,height=5.5cm]{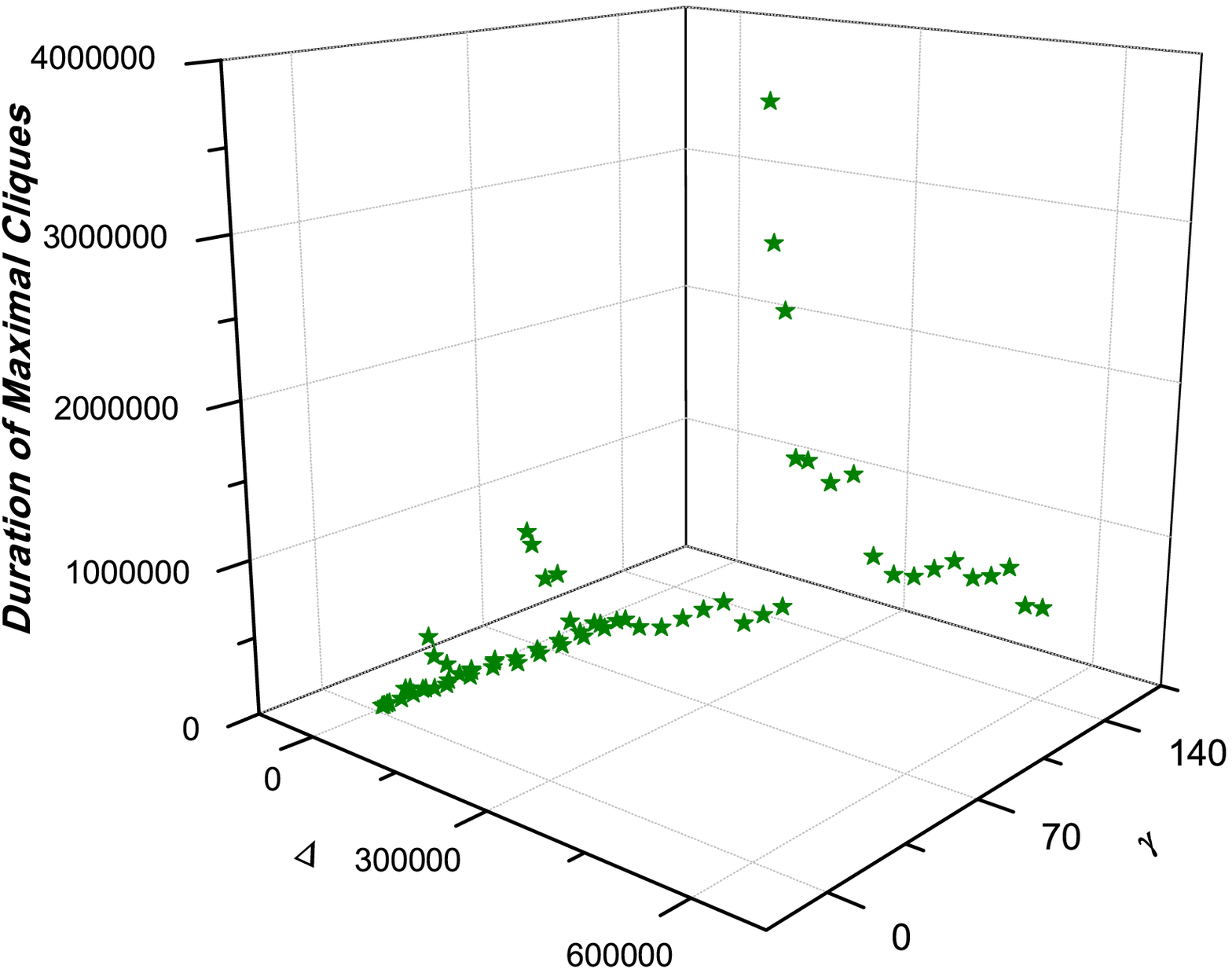}
 \\
(i-b)& (ii-b)&(iii-b)\\

  \includegraphics[width=6cm,height=5.5cm]{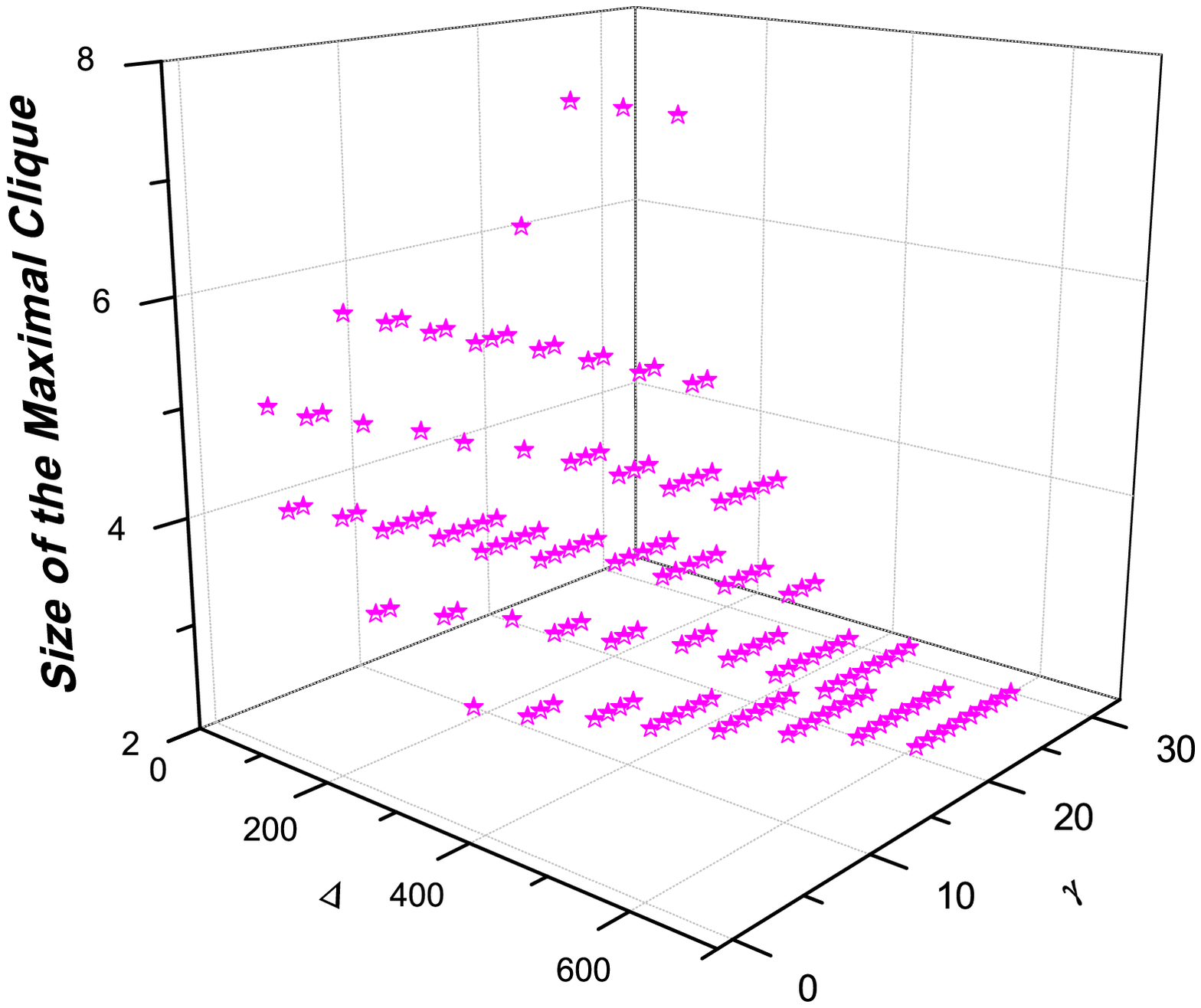}
   &
   \includegraphics[width=6cm,height=5.5cm]{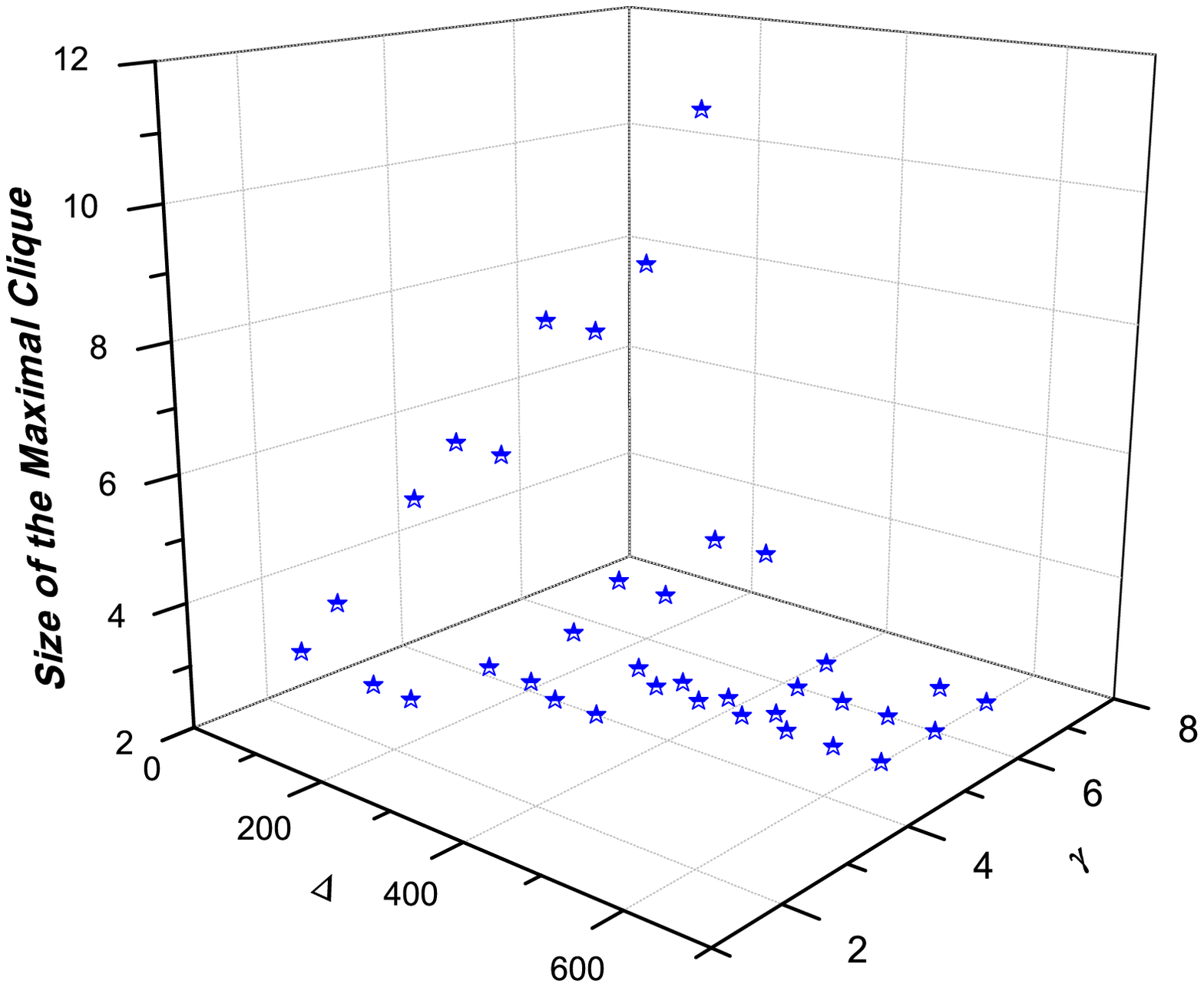}
   &
   \includegraphics[width=6cm,height=5.5cm]{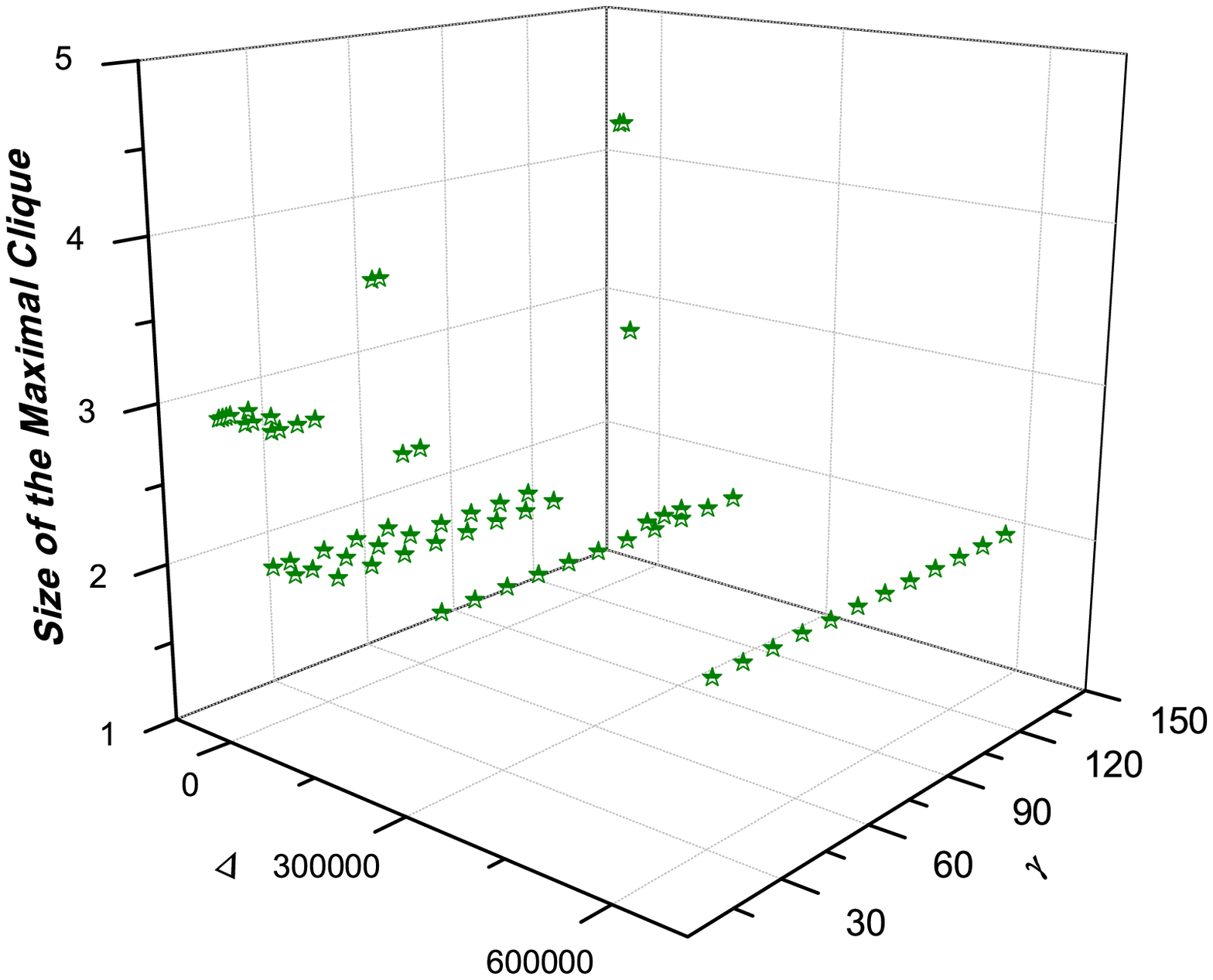}
  \\
  (i-c)& (ii-c)&(iii-c)\\

     \includegraphics[width=6cm,height=5.5cm]{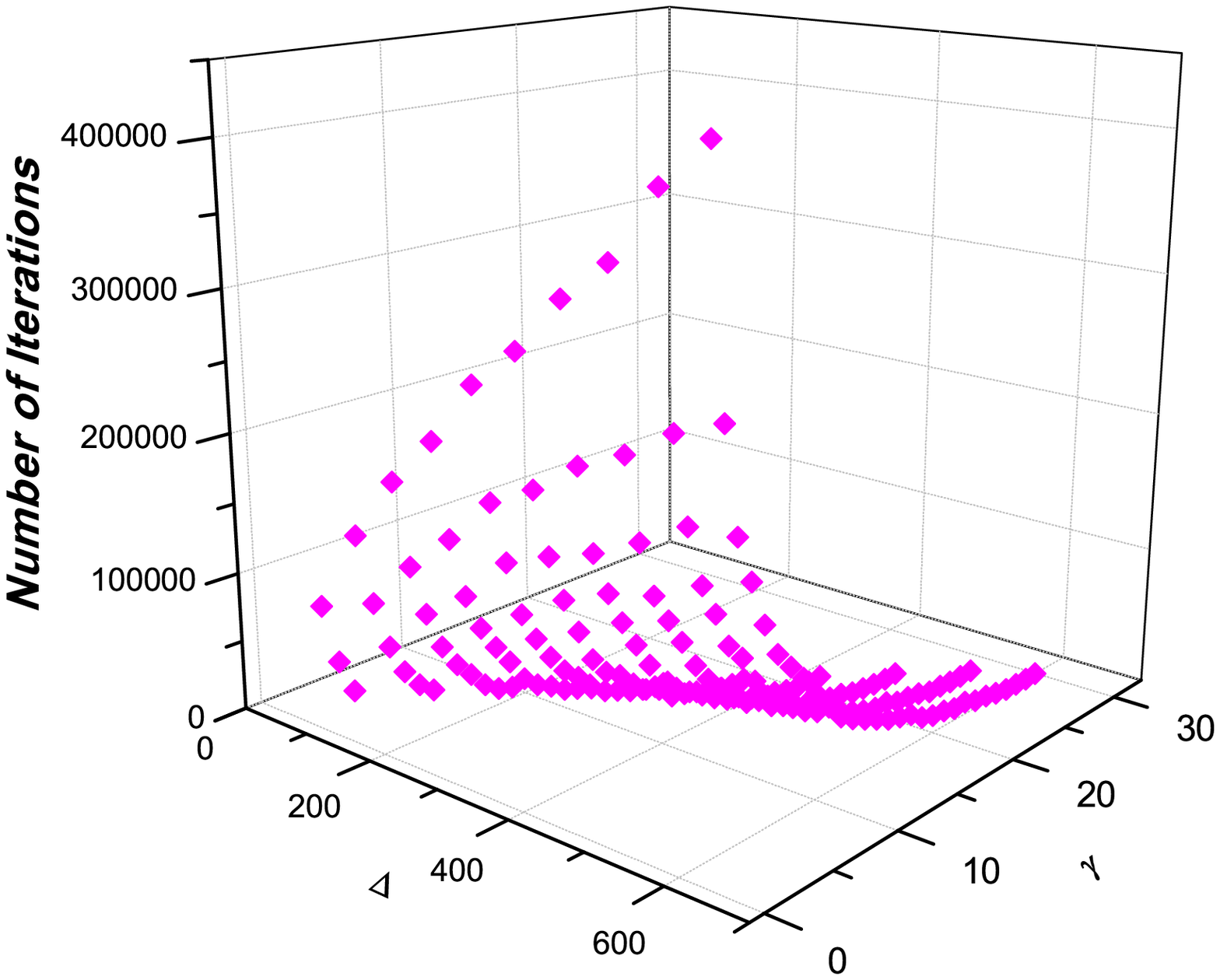} 
  &
  \includegraphics[width=6cm,height=5.5cm]{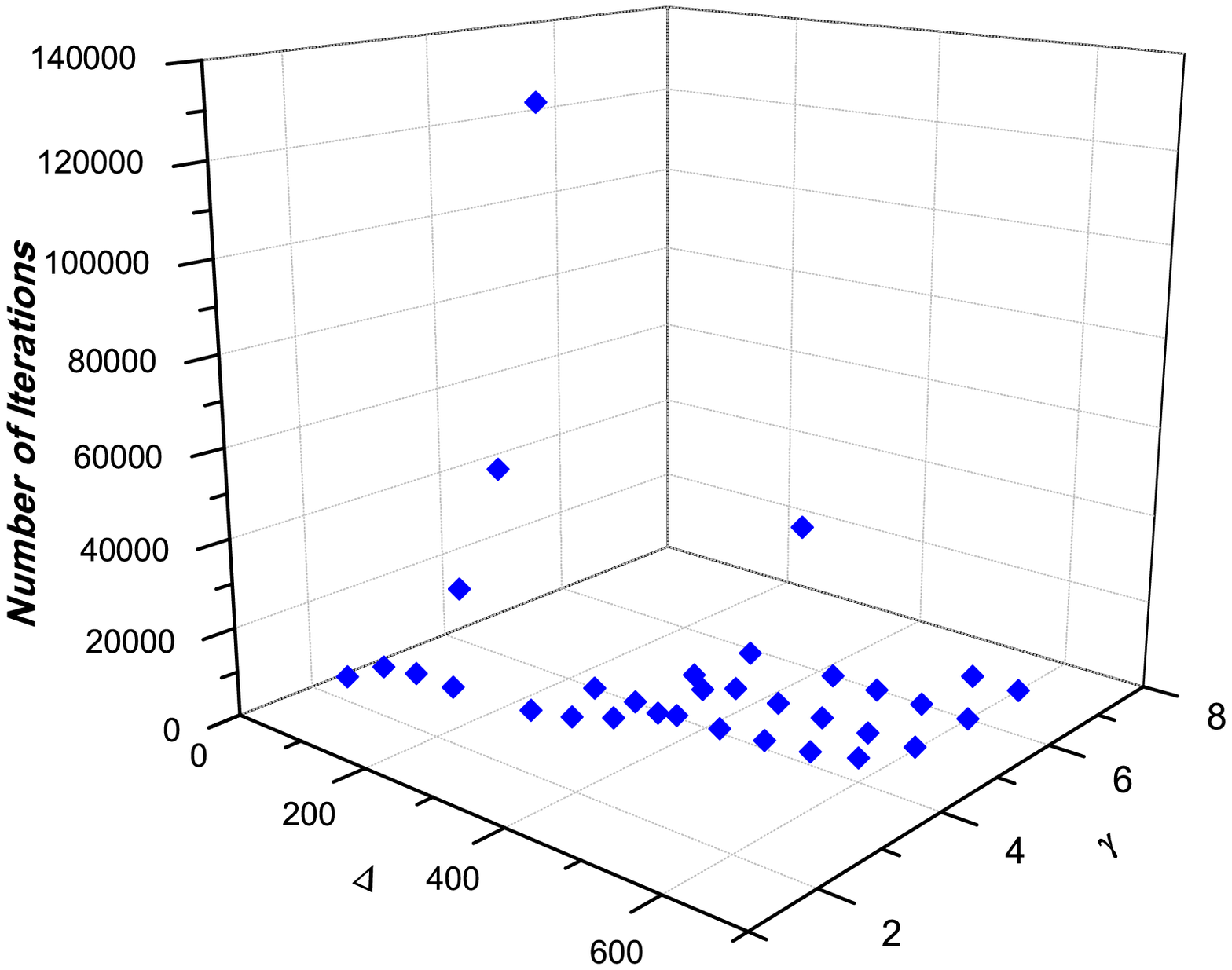} 
  &
  \includegraphics[width=6cm,height=5.5cm]{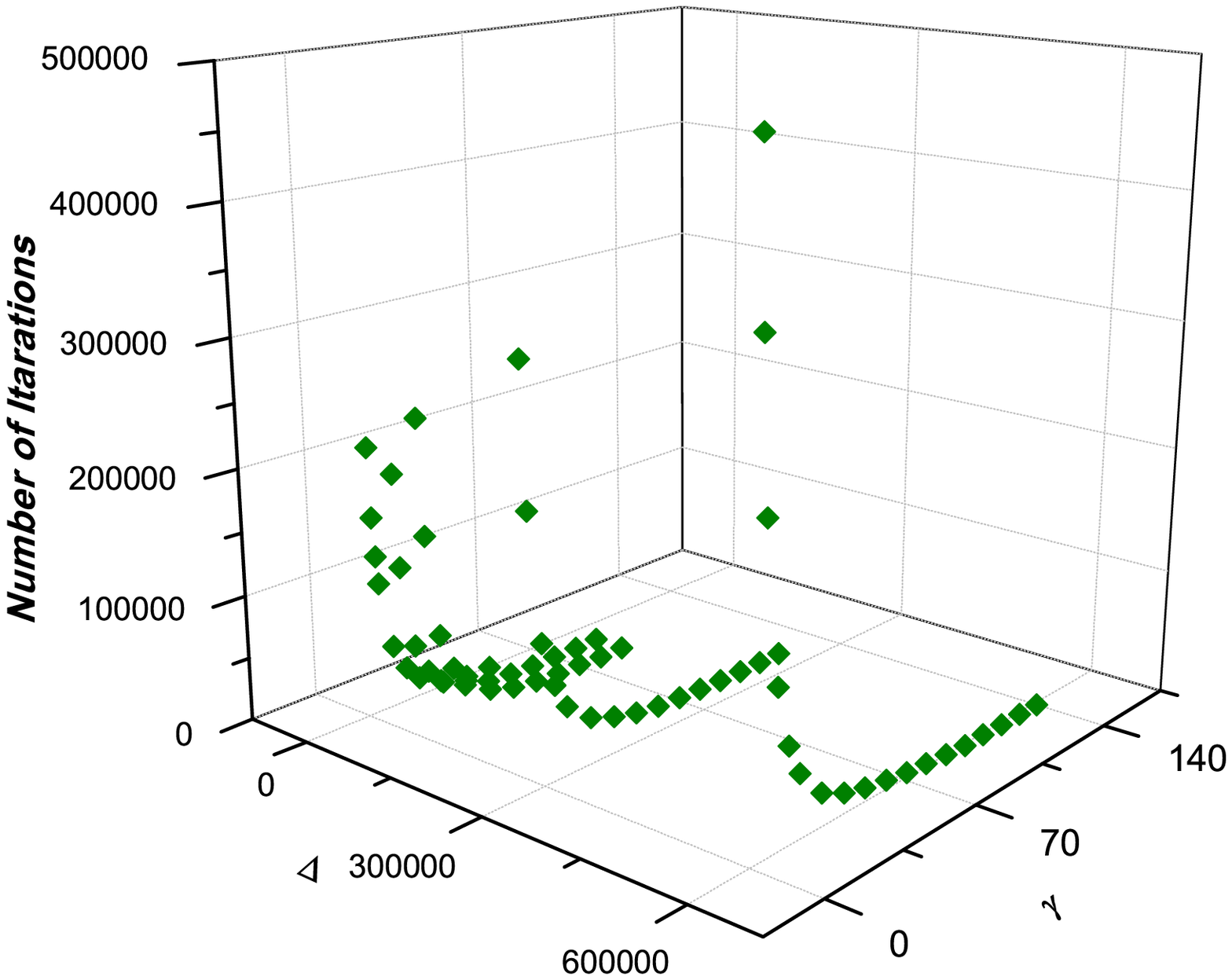} 
   \\

   (i-d)& (ii-d)&(iii-d)\\

\end{tabular}
\caption{Results obtained using Infectious Dataset in (i), Haggle Dataset in (ii) and College Message Dataset in (iii). Figure (a), (b), (c) and (d) respectively plots the Number of Maximal Cliques, Maximum Clique duration, maximum cardinality and Number of iterations in the maximal clique set for different $\Delta$, $\gamma$ values.}
\label{Fig:results}
\end{figure*}

\par For Infectious and College Message datasets, we observe that the maximal cliques are more in number when $\Delta$ and $\gamma$ values are lower. The similar pattern also exists in Haggle dataset. However, due to the abrupt change in number of maximal cliques with the increment of $\gamma$ in each $\Delta$, look wise Figure (ii-a) is slightly different from the other two. We also observe that for a particular $\Delta$, there is exponential decrease in number of maximal cliques with the growth of $\gamma$. Also, it is important to note, that for any $\Delta$, theoretically there will be no maximal $(\Delta, \gamma)$\mbox{-}Clique when $\gamma > \Delta+1$. The supportive observation is found in our experimentation. As an example, for Infectious dataset for $\Delta=300$, there is no maximal $(\Delta, \gamma)$\mbox{-}Clique beyond $\gamma=16$ (Each link has been captured in every 20 seconds. So, $\gamma$ can reach to maximum 300/20+1=16)

\par For all the datasets, larger $\Delta$ leads to increase in the maximum duration which is obvious. The maximum duration decreases with the increase of clique cardinality. For a fixed $\Delta$, the incremental change in $\gamma$ results in smaller value of the maximum duration. However, the change is not exponential as the previous metric (count of the maximal cliques). For Infectious and Haggle dataset, we observe similar plots in Figure \ref{Fig:results} (i-b) and (ii-b) respectively. In College Message dataset, due to selection of nonuniform $\Delta$ values, the change in maximum duration is very high in Figure (iii-b).

\par Next, we want to observe the maximum cardinality of the maximal cliques, which signifies at most how many persons contacted most frequently (preserving at least $\gamma$ times) for each $\Delta$ time interval. All the datasets exhibit following natural patterns, increase and decrease in the maximum cardinality along with the incremental change in $\Delta$ (for fixed $\gamma$) and $\gamma$ (for fixed $\Delta$) respectively. 

\par Now, the number of iterations to reach the maximal clique set is dependent on two parameters, one is the number of maximal cliques and second is the delta value for all the datasets. The plots in the first and fourth row of Figure \ref{Fig:results} are correspondingly almost identical with the change in $\gamma$ for a fixed $\Delta$ as it is comparable with number of maximal cliques. However, with the increase of $\Delta$ for fixed $\gamma$, there is almost linear increase in the number of iterations due to larger size of the initial clique set.

\section{Conclusions and Future Directions}\label{CFD}
In this paper, we have introduced the concept of $(\Delta, \gamma)$\mbox{-}Clique of a temporal network and proposed a methodology for enumerating all such maximal cliques. We have also used this methodology for analyzing three human contact network datasets captured in different situations. Now, this work can be extended in several directions. First of all, the analysis that we have done for our proposed methodology is not tight. Hence, a sophisticated analysis can be provided for the developed algorithm. Secondly, the the problem of enumerating maximal $(\Delta,\gamma)$\mbox{-}Cliques can be extended for uncertain graphs, where along with the time stamp, each edge also has a probability of occurrence. In reality, there could be different situations, where it might be interesting to study the interaction patterns among a group of objects for a particular duration with different level of frequencies. So, another direction of extension is to study the $(\Delta, \gamma)$\mbox{-}Clique where the context of $\gamma$ is subjective and case specific.

\bibliographystyle{alpha}
\bibliography{bare_conf}

\newcommand{\etalchar}[1]{$^{#1}$}
\begin{thebibliography}{HMNS16}

\bibitem[BK73]{bron1973algorithm}
Coen Bron and Joep Kerbosch.
\newblock Algorithm 457: finding all cliques of an undirected graph.
\newblock {\em Communications of the ACM}, 16(9):575--577, 1973.

\bibitem[CHC{\etalchar{+}}07]{chaintreau2007impact}
Augustin Chaintreau, Pan Hui, Jon Crowcroft, Christophe Diot, Richard Gass, and
  James Scott.
\newblock Impact of human mobility on opportunistic forwarding algorithms.
\newblock {\em IEEE Transactions on Mobile Computing}, 6(6):606--620, 2007.

\bibitem[CZKC12]{cheng2012fast}
James Cheng, Linhong Zhu, Yiping Ke, and Shumo Chu.
\newblock Fast algorithms for maximal clique enumeration with limited memory.
\newblock In {\em Proceedings of the 18th ACM SIGKDD international conference
  on Knowledge discovery and data mining}, pages 1240--1248. ACM, 2012.

\bibitem[ELS13]{eppstein2013listing1}
David Eppstein, Maarten L{\"o}ffler, and Darren Strash.
\newblock Listing all maximal cliques in large sparse real-world graphs.
\newblock {\em Journal of Experimental Algorithmics (JEA)}, 18:3--1, 2013.

\bibitem[ES11]{eppstein2011listing}
David Eppstein and Darren Strash.
\newblock Listing all maximal cliques in large sparse real-world graphs.
\newblock {\em Experimental Algorithms}, pages 364--375, 2011.

\bibitem[GJ02]{garey2002computers}
Michael~R Garey and David~S Johnson.
\newblock {\em Computers and intractability}, volume~29.
\newblock wh freeman New York, 2002.

\bibitem[HCM15]{hulovatyy2015exploring}
Yuriy Hulovatyy, Huili Chen, and T~Milenkovi{\'c}.
\newblock Exploring the structure and function of temporal networks with
  dynamic graphlets.
\newblock {\em Bioinformatics}, 31(12):i171--i180, 2015.

\bibitem[HMNS16]{himmel2016enumerating}
Anne-Sophie Himmel, Hendrik Molter, Rolf Niedermeier, and Manuel Sorge.
\newblock Enumerating maximal cliques in temporal graphs.
\newblock In {\em Advances in Social Networks Analysis and Mining (ASONAM),
  2016 IEEE/ACM International Conference on}, pages 337--344. IEEE, 2016.

\bibitem[HS12]{holme2012temporal}
Petter Holme and Jari Saram{\"a}ki.
\newblock Temporal networks.
\newblock {\em Physics reports}, 519(3):97--125, 2012.

\bibitem[HS13]{holme2013temporal}
Petter Holme and Jari Saram{\"a}ki.
\newblock {\em Temporal networks}.
\newblock Springer, 2013.

\bibitem[ISB{\etalchar{+}}11]{konect:sociopatterns}
Lorenzo Isella, Juliette Stehlé, Alain Barrat, Ciro Cattuto, Jean-François
  Pinton, and Wouter~Van den Broeck.
\newblock What's in a crowd? analysis of face-to-face behavioral networks.
\newblock {\em J. of Theoretical Biology}, 271(1):166--180, 2011.

\bibitem[JYP88]{johnson1988generating}
David~S Johnson, Mihalis Yannakakis, and Christos~H Papadimitriou.
\newblock On generating all maximal independent sets.
\newblock {\em Information Processing Letters}, 27(3):119--123, 1988.

\bibitem[kon17]{konect:2017:sociopatterns-infectious}
Infectious network dataset -- {KONECT}, April 2017.

\bibitem[MH17]{masuda2017temporal}
Naoki Masuda and Petter Holme.
\newblock {\em Temporal Network Epidemiology}.
\newblock Springer, 2017.

\bibitem[MXT15]{mukherjee2015mining}
Arko~Provo Mukherjee, Pan Xu, and Srikanta Tirthapura.
\newblock Mining maximal cliques from an uncertain graph.
\newblock In {\em Data Engineering (ICDE), 2015 IEEE 31st International
  Conference on}, pages 243--254. IEEE, 2015.

\bibitem[MXT17]{mukherjee2017enumeration}
Arko~Provo Mukherjee, Pan Xu, and Srikanta Tirthapura.
\newblock Enumeration of maximal cliques from an uncertain graph.
\newblock {\em IEEE Transactions on Knowledge and Data Engineering},
  29(3):543--555, 2017.

\bibitem[POC09]{panzarasa2009patterns}
Pietro Panzarasa, Tore Opsahl, and Kathleen~M Carley.
\newblock Patterns and dynamics of users' behavior and interaction: Network
  analysis of an online community.
\newblock {\em Journal of the Association for Information Science and
  Technology}, 60(5):911--932, 2009.

\bibitem[VLM15]{viard2015revealing}
Jordan Viard, Matthieu Latapy, and Cl{\'e}mence Magnien.
\newblock Revealing contact patterns among high-school students using maximal
  cliques in link streams.
\newblock In {\em Proceedings of the 2015 IEEE/ACM International Conference on
  Advances in Social Networks Analysis and Mining 2015}, pages 1517--1522. ACM,
  2015.

\bibitem[VLM16]{viard2016computing}
Tiphaine Viard, Matthieu Latapy, and Cl{\'e}mence Magnien.
\newblock Computing maximal cliques in link streams.
\newblock {\em Theoretical Computer Science}, 609:245--252, 2016.

\end{thebibliography}

\end{document}